\documentclass[11pt]{article}

\usepackage{amssymb,amsfonts,xspace,graphicx,relsize,bm,mathtools,xcolor,amsmath,breqn,algpseudocode,mathtools,mathrsfs,multirow}
\usepackage{calrsfs}
\usepackage[T1]{fontenc}
\usepackage[utf8]{inputenc}
\usepackage[pagebackref]{hyperref}
\usepackage[usestackEOL]{stackengine}
\usepackage[thmmarks]{ntheorem}

\usepackage{cleveref}
\usepackage{thm-restate}
\usepackage{array}
\usepackage{parskip}

\newcommand{\rotleq}{\rotatebox[origin=c]{90}{$\leq$}}
\newcommand{\roteq}{\rotatebox[origin=c]{90}{=}}
\newcommand{\rotleqstar}{\rotatebox[origin=c]{90}{$\stackrel{\star}{\leq}$}}

\usepackage{url}
\def\01{\{0,1\}}

\newcommand{\ket}[1]{|#1\rangle}

\newcommand{\ketbra}[2]{|#1\rangle\langle#2|}

\newcommand{\Tr}{\mbox{\rm Tr}}

\usepackage{ntheorem}

\DeclareMathOperator{\poly}{poly}

\newtheorem{theorem}{Theorem}[section]
\newtheorem{definition}[theorem]{Definition}
\newtheorem{fact}[theorem]{Fact}

\newtheorem{lemma}[theorem]{Lemma}

\newtheorem{corollary}[theorem]{Corollary}

\newtheorem{claim}[theorem]{Claim}
\def\01{\{0,1\}}

\usepackage{caption}

\usepackage{array}
\newcolumntype{M}[1]{>{\centering\arraybackslash}m{#1}}

\usepackage[margin=1in]{geometry}
\hypersetup{
	colorlinks,
	linkcolor={red!100!black},
	citecolor={blue!100!black},
}
\usepackage{comment}
\usepackage{multicol}

\newenvironment{proof}
{\noindent {\bf Proof. }}
{{\hfill $\Box$}\\
	\smallskip}

\algtext*{EndFor}

\renewcommand\thmcontinues[1]{Formal Statement}

\usepackage{footnote}

\DeclareMathAlphabet{\pazocal}{OMS}{zplm}{m}{n}

\newcommand{\Int}{\ensuremath{\mathsf{Int}}}
\newcommand{\Nequality}{\ensuremath{\mathsf{pDist}}}
\newcommand{\GHD}{\ensuremath{\mathsf{GHD}}}
\newcommand{\GH}{\ensuremath{\mathsf{GH}}}
\newcommand{\SM}{\ensuremath{\mathsf{S}}}

\usepackage{calrsfs}

\newcommand{\Mset}{\ensuremath{\mathcal{M}}}
\newcommand{\NM}{\ensuremath{\mathsf{NM}}}

\newcommand{\M}{\ensuremath{\mathsf{M}}}

\newcommand{\QM}{\ensuremath{\mathsf{QM}}}

\newcommand{\QNM}{\ensuremath{\mathsf{QNM}}}
\newcommand{\EQ}{\ensuremath{\mathsf{EQ}}}
\newcommand{\MAJ}{\ensuremath{\mathsf{MAJ}}}
\newcommand{\DISJ}{\ensuremath{\mathsf{DISJ}}}
\usepackage{parskip}
\def\01{\{0,1\}}
\newcommand{\GBBP}{\mathsf{GBBP}}
\newcommand{\BP}{\mathsf{BP}}
\newcommand{\BBP}{\mathsf{BBP}}
\newcommand{\parity}{\mathsf{PARITY}}
\newcommand{\IP}{\mathsf{IP}}

\usepackage[noline, boxed]{algorithm2e}

\begin{document}

\title{Communication memento:\\ Memoryless communication complexity}

\author{
Srinivasan Arunachalam\\ 
\small IBM Research, USA\\
\small \texttt{Srinivasan.Arunachalam@ibm.com}
\and
Supartha Podder\\
\small University of Ottawa, Canada\\
\small \texttt{spodder@uottawa.ca}
}

\date{}

\maketitle

\begin{abstract}
We study the communication complexity of computing functions $F~:~\01^n\times~\01^n\rightarrow \01$ in the \emph{memoryless communication model}.  Here, Alice is given $x\in \01^n$, Bob is given $y\in \01^n$ and their goal is to compute $F(x,y)$ subject to the following constraint: at every round, Alice receives a message from Bob and her reply to Bob solely depends on the message received and her input $x$ (in particular, her reply is independent of the information from the previous rounds); the same applies to Bob.  
 The cost of computing~$F$ in this model is the \emph{maximum} number of bits exchanged in any round between Alice and Bob (on the worst case input $x,y$). 
 In this paper, we also consider variants of our memoryless model wherein one party is allowed to have memory, 
 the parties are allowed to communicate quantum bits, 
 only one player is allowed to send messages. 
We show that some of these different variants of our memoryless communication model capture the garden-hose model of computation by Buhrman et al.~(ITCS'13), space bounded communication complexity by Brody et al.~(ITCS'13) and the overlay communication complexity by Papakonstantinou et al.~(CCC'14). Thus the memoryless communication complexity model provides a unified framework to study all these space-bounded communication complexity models.

 We establish the following main results: (1) We show that the memoryless communication complexity of $F$ equals the logarithm of the size of the smallest bipartite branching program computing $F$ (up to a factor $2$); (2) We show that memoryless communication complexity equals garden-hose model of computation; 
 (3) We exhibit various exponential separations between these memoryless communication models.

We end with an intriguing open question: can we find an explicit function $F$ and universal constant $c>1$ for which the memoryless communication complexity is at least $c \log n$? Note that 
$c\geq 2+\varepsilon$ would imply a $\Omega(n^{2+\varepsilon})$ lower bound for general formula size, improving upon the best lower bound by Ne\v{c}iporuk~\cite{nechiporuk:cc}.
\end{abstract}

\section{Introduction}
Yao~\cite{yao:comm} introduced the model of \emph{communication complexity} in 1979 and ever since it's introduction, communication complexity has played a pivotal role in understanding various problems in theoretical computer science. In its most general form in this model, the goal is the following: there are two separated parties usually referred to as Alice and Bob, Alice is given an $n$-bit string $x\in \01^n$ and similarly Bob is given $y\in \01^n$ and together they want to compute $F(x,y)$ where $F:\01^n\times \01^n\rightarrow \01$ is a function known to both of them. Here Alice and Bob are given unlimited computational time and memory and the \emph{cost} of any communication protocol between Alice and Bob is the \emph{total number of bits exchanged} between them. Clearly a trivial protocol is Alice sends her input $x$ to Bob who can then compute $F(x,y)$, which takes $n$ bits of communication. Naturally, the goal in communication complexity is to minimize the number of bits of communication between them before computing $F(x,y)$.  The \emph{deterministic communication complexity} of a function~$F$ (denoted $\textsf{D}(F)$) is defined as the total number of bits of communication before they can decide $F(x,y)$ on the worst-case inputs~$x,y$. 

Since its introduction there have been various works that have extended the standard {deterministic} communication model to the setting where Alice and Bob are allowed to share randomness and need to output $F(x,y)$ with high probability (probability taken over the randomness in the protocol). Apart from this there has been studies on non-deterministic communication complexity~\cite{de2003nondeterministic}, quantum communication complexity~\cite{yao1993quantum} (wherein Alice and Bob are allowed to share \emph{quantum bits} and possibly have shared entanglement), \emph{unbounded error communication complexity}~\cite{paturi1986probabilistic} and their variants. One-way variants have also been considered where only Alice sends messages to Bob.
Study of these different models of communication complexity and their variants have provided many important results in the fields of VLSI~\cite{Pal:ccandparral}, circuit lower bounds~\cite{goldmanhastad:cc}, algorithms~\cite{AMS:sketching}, data structures~\cite{ccfordata}, property testing~\cite{blaisbrodymatulef:property}, streaming algorithms~\cite{bar2004information}, computational complexity~\cite{ccforcomputationalcomplexity}, extended formulations~\cite{fiorini:extended}.\footnote{For more on communication complexity and its applications, we refer the interested reader to the standard textbooks for communication complexity~\cite{ccbook,lee:ccbook2}.}

\subsection{Background}
\paragraph{Space bounded communication complexity.} In the context of our current understanding of computation, the study of space required to solve a problem is 
a central topic in complexity theory. Several space bounded models such as width-bounded branching programs~\cite{lam1989tradeoffs}, limited depth circuits, straight line protocols~\cite{lam1992trade} have been widely studied in this context. In this direction variants of communication complexity have also been analyzed to better understand communication-space trade-offs~\cite{impagliazzo2010communication, klauck2004quantum, lam1989tradeoffs}. In particular, the relation between space-bounded computation and communication complexity was formally initiated by Brody et al.~\cite{brody2013space} who considered the following question: what happens if we change the standard communication model such that, in each step of communication, Alice and Bob are limited in their ability to store the information from the previous rounds (which includes their private memory and messages exchanged). 
In this direction, they introduced a new model wherein Alice and Bob each are allowed to store at most $s(n)$ bits of memory and showed that unlike the standard communication complexity, in this model \emph{super-linear} lower bounds on the amount of communication is possible.\footnote{We remark that the separations obtained by~\cite{brody2013space} were for non-Boolean functions.} 
Brody et al.\ mainly studied one-way communication complexity variant of this limited memory model in which Bob can have two types of memory: an oblivious memory (depends only on Alice's message) and a non-oblivious memory  (for computation). With these definitions, they obtained memory hierarchy theorems for such communication models analogous to the space hierarchy theorem in the Turing machine world.  

\paragraph{Overlay communication complexity.} Subsequently, Papakonstantinou, et al.~\cite{memoryless:comm}  defined a similar space-bounded \emph{one-way} communication model wherein Alice has unlimited memory and Bob has either no memory or constant-sized memory. At each round, messages from Alice to Bob consists of at most $t(n)$ bits and the complexity of computing $F$ is the maximum $t(n)$ required over all inputs to $F$. They characterized the complexity in this model by an elegant combinatorial object called the \emph{rectangle overlay} (which is defined in Section~\ref{subsec:overlay}). They also managed to establish connections between their model and the well-known communication complexity polynomial hierarchy, introduced by Babai, Frankl and Simon \cite{babai1986complexity}. Papakonstantinou et al.~\cite{memoryless:comm} showed that the message length in their model corresponds to the oblivious memory in a variant of space bounded model, introduced by Brody et al.~\cite{brody2013space}, where Bob only has access to an oblivious~memory.

\paragraph{Garden-hose model.} Another seemingly unrelated complexity model, the garden-hose complexity was introduced by Buhrman et al.~\cite{gardenhose:intro}  to understand quantum attacks on position-based cryptographic schemes (see Section~\ref{sec:garden} for a formal definition). Polynomial size garden-hose complexity is known to be equivalent to Turing machine $\log$-space computations with pre-processing. In the garden-hose model two distributed players Alice and Bob use several pipes to send water back and forth and compute Boolean functions based on whose side the water spills. Garden-hose model was shown to have many connections to well-established complexity models like formulas, branching programs and circuits. 
A long-standing open question in this area is, is there an explicit function on $n$ bits whose garden-hose complexity is super-linear in $n$?

\paragraph{Branching programs.}  Another unrelated computation model is branching program. Understanding the size of De Morgan formulas that compute Boolean functions has a long history. In particular, there has been tremendous research in understanding lower bounds on size of De Morgan formulas computing a Boolean function.
Similar to formulas, branching programs have also been well-studied in complexity theory. For both branching programs and formulas, we have explicit functions which achieve quadratic (in input size) lower bounds on the size of the branching program/formula computing them. 
A few years ago, Tal~\cite{tal:quadraticbipartite} considered \emph{bipartite formulas} for $F:X\times Y\rightarrow \01$ (where each internal node computes an arbitrary function on either $X$ or $Y$, but not both) and showed that the inner product function requires quadratic-sized formulas to compute. In the same spirit as Tal's result, a natural open question is, is there an explicit bipartite function which requires super-linear sized bipartite branching programs to compute?

Given these different models of computation, all exploring the effects on computation under various restrictions, a natural question is, can we view all of them in a unified~way:

\begin{question*}
\label{open:main-question}
Is there a model of communication that captures all the above computational models?
\end{question*}
%}

\textbf{In this work} we introduce a very simple and new framework called \emph{the memoryless communication complexity} which captures all the computational models mentioned above.

\subsection{Memoryless Communication Models}

\textbf{Memoryless communication models.}
We introduce a \emph{natural} model of communication complexity which we call \emph{the memoryless communication complexity}. Here, like the standard communication complexity, there are two parties Alice and Bob given $x,y$ respectively and they need to compute $F(x,y)$, where $F:\01^n\times \01^n\rightarrow \01$ is known to both of them. However, we tweak the standard communication model in the following two ways: The first change is that Alice is ``\emph{memoryless}'', i.e.,  at every round Alice computes the next message to send solely based on only her input~$x$ and the message received from Bob in this round.
She does not remember the entire transcript of messages that were communicated in the previous rounds and also forgets all the private computation she did in the previous rounds.  
Similarly Bob computes a message which he sends to Alice, based \emph{only} on his input $y$ and the message received from Alice in the current round. After Bob sends his message, he also forgets the message received and all his private computations. Alice and Bob repeat this procedure for a certain number of rounds before one of them outputs~$F(x,y)$. 

The second crucial change in the memoryless communication model is that the \emph{cost} of computing~$F$ in this model is the \emph{size of the largest message} communicated between Alice and Bob in any round of the protocol (here size refers to the number of bits in the message). Intuitively, we are interested in knowing what is the size of a re-writable message register (passed back and forth between Alice and Bob) sufficient to compute a function~$F$ on all inputs $x$ and~$y$, wherein Alice and Bob do not have any additional memory to remember information between rounds. We denote the \emph{memoryless communication cost of computing $F$} as $\NM(F)$ (where $\NM$ stands for ``no-memory''). We believe this communication model is very natural and as far as we are aware this memoryless communication model wasn't defined and studied before in the classical literature. 

Being more formal, we say $F:\01^n\times \01^n\rightarrow \01$ can be computed in the memoryless communication model with complexity $t$, if the following is true. For every $x,y\in \01^n$ there exists functions $\{f_x,g_y:\01^t\rightarrow \01^t\}$ such that, on input $x,y$, Alice and Bob use $f_x$ and $g_y$ respectively to run the following protocol: the first message in the protocol is $f_x(0^t)$ from Alice to Bob and thereafter, for every message $m_B$ Bob receives, he replies with deterministic $m'=g_y(m_B)$ and similarly for every message $m_A$ Alice receives shes replies with $m''=f_x(m_A)$. The protocol terminates when the transcript is $(1^{t-1}b)$ at which point they output $b$ as their guess for $F(x,y)$; and we say the protocol computes $F$ if for every $x,y$, the output $b$ equals $F(x,y)$.  $\NM(F)$ is defined as the smallest $t$ that suffices to compute $F$ (using the protocol above) for every $x,y\in \01^n$.

It is worth noting that in the memoryless communication protocol, Alice and Bob do not even have access to clocks and hence cannot tell in which round they are in (without looking at the message register). Hence, every memoryless protocol can be  viewed  as Alice and Bob applying \emph{deterministic} functions (depending on their inputs) which map incoming messages to out-going messages. Also note that unlike the standard communication complexity, where a single bit-message register suffices for computing all functions (since everyone has memory), in the $\NM$ model because of the memoryless-ness we need more than a single bit register for computing almost all functions.

For better understanding, let us look at a protocol for the standard \emph{equality} function defined as $\EQ_n:\01^n\times \01^n\rightarrow \01$ where $\EQ_n(x,y)=1$ if and only if $x=y$. It is well-known that $\textsf{D}(\EQ_n)=n$. In our model, we show that $\NM(\EQ_n)\leq \log n+1$: for $i=1,\ldots,n$, at the $i$th round, Alice sends the $(\log n+1)$-bit message $(i,x_i)$ and Bob returns $(i,[x_i=y_i])$,\footnote{Here $[\cdot]$ is the indicator of an event in the parenthesis.} Alice increments $i$ and repeats this protocol for $n$ rounds. In case Bob finds an $i$ for which $x_i\neq y_i$, he outputs $0$, if not after $n$ rounds they output $1$. Note that this protocol didn't require Alice and Bob to have any memory and the length of the longest message in this protocol was $\log n+1$.  We discuss more protocols later in the paper and formally describe the memoryless communication model in~Section~\ref{sec:memcommmode}. 

\paragraph{Variants of the memoryless model.} Apart from the memoryless communication complexity, we will also look at the ``{memory-nomemory communication complexity}'' wherein Alice is allowed to have memory (i.e., Alice can know which round she is in, can remember the entire transcript and her private computations of each round) whereas Bob doesn't have any memory during the protocol. The goal of the players remains to compute a function $F$ and the \emph{cost} of these protocols (denoted by $\M(F)$) is still defined as the smallest size of a message register required between them on the worst inputs. Apart from this, we will also consider the quantum analogous of these two communication models wherein the only difference is that Alice and Bob are allowed to send \emph{quantum} bits. We formally describe these models of communication in Section~\ref{sec:memcommmode}. In order to aid the reader we first set up some notation which we use to describe our results: for $F:\01^n\times \01^n\rightarrow \01$, let
\begin{enumerate}
    \item $\NM(F)$ be the memoryless communication complexity of computing $F$ wherein Alice and Bob both do not have any memory.
    \item $\M(F)$ be the memory-nomemory communication complexity of computing $F$ where Alice has memory and Bob doesn't have memory
\end{enumerate}
Apart from these, we will also allow quantum bits of communication between Alice and Bob and the complexities in these models are denoted by $\QNM(F)$ and $\QM(F)$. Additionally, we will consider the one-way communication variants wherein only Alice can send messages to Bob and the complexities in these models are denoted by $\NM^{\rightarrow}(F), \M^{\rightarrow}(F), \QNM^{\rightarrow}(F), \QM^{\rightarrow}(F)$.

\subsection{Our Contributions}

 The main contribution in this paper is to first define the model of the memoryless communication complexity and consider various variants of this model (only some of which were looked at before in the literature). We emphasize that we view our main contribution as a new \emph{simple} communication model that provides a conceptual – rather than technical – contribution to studying space complexity, bipartite branching programs and garden-hose complexity under a single model. Given the vast amount of research in the field of communication complexity, we believe that our memoryless model is very natural model of computation. We now state of various connections between our memoryless communication model and other computational models.

\paragraph{1. Characterization in terms of branching programs.}
It is well-known that standard models of communication complexity are characterized by the so-called (monochromatic) ``rectangles'' that partition the communication matrix of the function Alice and Bob are trying to compute. In the study of memoryless model, Papakonstantinou et al.~\cite{memoryless:comm} specifically consider the memory-nomemory model of communication complexity wherein Alice has a memory and Bob doesn't and they are restricted to one-way communication from Alice to Bob. They show a beautiful combinatorial rectangle-overlay characterization (denoted $\textsf{RO}(F)$) of the  memory-no memory communication model.\footnote{This rectangle-overlay complexity is formally defined in Section~\ref{subsec:overlay}.}
 One natural idea is to improve the $\textsf{RO}(F)$ complexity to a more fine-grained rectangle measure that could potentially also characterize $\NM(F)$, but this doesn't seem to be true. The fact that both Alice and Bob do not have memory, doesn't allow them to ``narrow'' down into a rectangle allowing them to compute the function, instead they narrow down into a set of rectangles. This motivates the question, is there a natural characterization of even the memoryless communication model, in which both Alice and Bob do not have memory? Here we answer this in the positive. We provide a characterization of memoryless communication complexity using branching programs. In particular, we show that for every $F:\01^n\times \01^n\rightarrow \01$, the memoryless complexity $\NM(F)$ is (up to a factor $2$) equal to the logarithm of the size of the smallest bipartite branching program computing $F$.\footnote{We defer the formal definition of such branching programs to Section~\ref{sec:prelim} and Section~\ref{subsec:overlay}.} We give a proof of this statement in Theorem~\ref{thm:d_nomem_logT}.

\paragraph{2. Characterization in terms of garden hose complexity} 
	The garden-hose model of computation was introduced by Buhrman et al.~\cite{gardenhose:intro} to understand quantum attacks on position-based cryptographic schemes. It is a playful communication model where two players compute a function with set of pipes, hoses and water going back-and-forth through them. 
	Alice and Bob start with $s$ pipes and based on their private inputs ``match'' some of the openings of the pipes on their respective sides. Alice also connects a water tap to one of the open pipes. Then based on which side the water spills they decide on the function value. Naturally they want to minimize the number of pipes required over all possible inputs and the garden-hose complexity $\GH(F)$ is defined to be the minimum \emph{number of pipes} required to compute $F$ this way.  
	 Given it's puzzle-like structure, there have been several works to understand this model and various connections between the garden-hose model and other branches of theoretical computer science were established~\cite{gardenhose:intro,podder:garden,speelman:thesis,buhrman2016quantum,chiu2014garden,speelman:thesis2,dulek2016quantum}.

On the other hand, space bounded communication complexity was introduced by Brody et al.~\cite{brody2013space} to study the effects on communication complexity when they players are limited in their ability to store information from previous rounds. Here Alice and Bob each have at most $s(n)$ bits of memory. Based on their private inputs $x,y$ they want to compute the function in a manner in which at each round Alice receives a single bit message from Bob and based on her input $x$, the incoming message $m_B$ and her previous $s(n)$-bit register content, she computes a new~$s(n)$-bit register and decides whether to stop and output $0/1$ or to continue. Bob does the same.  
Space bounded communication complexity $\SM(F)$ of computing a function~$F$ is the minimum register size~$s(n)$ required to compute the function on the hardest input.    
        
It was already shown by \cite{brody2013space} that for every function, the logarithm of the garden-hose complexity and the space bounded communication complexity is factor 2 related. It is also easy to show that our newly defined memoryless communication complexity is also factor 2 related to the space bounded communication complexity by \cite{brody2013space}. 

\[\NM(F) \leq 2\cdot \SM(F)+1, \text{ and } \hspace{2mm} \SM(F) \leq \NM(F) + \log \NM(F)\]
 We give a proof of this statement in Lemma~\ref{label:lemma-nm-sm}. Thus it immediately follows that the logarithm of the garden-hose complexity and the memoryless communication complexity of any function is also at most factor 3 related. 
 However we improve this relation using an elegant trick of~\cite{lange2000reversible} that allows one to make computations reversible; and thereby show that for every function $F$, $\NM(F)$ and $\GH(F)$ are equal up to an additive term 4. 

$$
\log\GH(F) -4 \leq \NM(F)\leq \log \GH(F).
$$
We give a proof of this in Theorem~\ref{lem:nmghm}. Hence, the memoryless communication complexity model provides a clean framework for studying all these apparently different looking computational~models. 

As an immediate application of this new characterization of the garden-hose model, we get a better upper bound for the garden-hose complexity of the indirect storage access function. 

\begin{definition}[Indirect Storage Access]
Let $n\geq 4$ and $m\geq 2$ be such that $n=2m+\log\big(m/(\log m)\big)$. The Indirect storage access function $\mathsf{ISA}_n : \{0,1\}^n \rightarrow \{0,1\}$ is defined on the input string $( \vec{x}, \vec{y}_1, \ldots, \vec{y}_{2^k}, \vec{z})$ where $k= \log \Big(\frac{m}{\log m}\Big)$, $\vec{x}\in \01^m$, $\vec{y}\in \01^{\log m}$, $\vec{z}\in \01^k$.
%Also, for all $i\in [m]$, $x_i\in \{0,1\}$ and for all $j\in [k]$, $z_j\in \{0,1\}$.
%We view the $\vec{y}$ as $\frac{m}{\lceil\log m\rceil}$ blocks of $\lceil \log m \rceil$ bits each.
%Let $0 \leq i \leq 2^k - 1$ be the integer encoded by the bits $\vec{y}$ and let $k_i$ be the integer encoded by the $i$-th block of log m bits in $\overrightarrow{y}$. Then $\mathsf{ISA}_n(\overrightarrow{x},\overrightarrow{y}, \overrightarrow{z}) = x_{y_z}$.
Then $\mathsf{ISA}_n( \vec{x}, \vec{y}_1, \cdots, \vec{y}_{2^k}, \vec{z})$ is evaluated as follows: compute $a=\Int(z)\in [2^k]$, then compute $b=\Int(\vec{y}_{a})\in [m]$ and output~$x_b$.

For the communication complexity version Alice gets $\vec{x},\vec{z}$ and Bob gets  $\vec{y}_1,\ldots,\vec{y}_{2^k}$, they want to compute~$\mathsf{ISA}_n(\vec{x},\vec{y}_1,\ldots,\vec{y}_{2^k},\vec{z})$.
\end{definition}

It was conjectured in \cite{podder:garden} that the Indirect Storage Access function has garden hose complexity $\Omega(n^2)$. This function is known to have $\Omega(n^2)$ lower bound for the branching program \cite{wegener1987complexity} and thus is believed to be hard for garden-hose model in general.\footnote{In a typical garden-hose protocol for computing $\mathsf{ISA}_n$, Alice uses $m$ pipes to describe $\vec{z}$ to Bob (each pipe for a different value of $\vec{z}$). Bob can then use another set of $m$ pipes to send $\vec{y}_{\Int(z)}$ to Alice. But since $\vec{y}_i$s need to be unique it appears that Bob needs $m$ set of such $m$ pipes in the worst case. This many-to-one mapping seems hard to tackle in the garden-hose model in general. Hence $\mathsf{ISA}_n$ appears to have an $\Omega(n^2)$ garden-hose complexity.} 
But it is easy to see that $\NM(\mathsf{ISA}) \leq \log n$: Alice sends $\vec{z}$ to Bob who then replies with $\vec{y}_{\Int(z)}$. Finally Alice computes the output. 
Thus using the memoryless-garden-hose equivalence (in Theorem~\ref{lem:nmghm}) we immediately get $\GH(\mathsf{ISA}) \leq 16n$ (thereby refuting the conjecture of~\cite{podder:garden}).

\paragraph{4. Separating these models.} We then establish the following inequalities relating the various models of communication complexity.\footnote{Some of the inequalities are straightforward but we explicitly state it for completeness.}
 
		\[
\begin{array}[t]{@{}c@{}}
      {\M(F)} \\
      \rotleq \\
      \QM(F)
  \end{array}
  \leq \begin{array}[t]{@{}c@{}}
      {\NM(F)} \\
      \rotleqstar \\
      \QNM(F)
  \end{array} 
  = \begin{array}[t]{@{}c@{}}
      {\log \GH(F)} \\
      \roteq \\
      2\cdot \textsf{S}(F)
  \end{array}   
    \leq \begin{array}[t]{@{}c@{}}
      {\M^{\rightarrow}(F)} \\
      \rotleq \\
      \QM^{\rightarrow}(F)
  \end{array} 
    \leq \begin{array}[t]{@{}c@{}}
      {\NM^{\rightarrow}(F)} \\
      \rotleq \\
      \QNM^{\rightarrow}(F)
  \end{array} 
\]
Furthermore, except the inequality marked by $(\star)$, we show the existence of various functions $F:\01^n\times \01^n\rightarrow \01$ for which every inequality is exponentially weak. In order to prove these exponential separations we use various variants of well-known functions such as inner product, disjointness, Boolean hidden matching problem, gap-hamming distance~problem. Giving exponential separations between quantum and classical communication complexity\footnote{These exponential separations are in the standard communication model where the communication complexity is the \emph{total} number of bits or qubits exchanged between Alice and Bob.} is an extensively studied subject~\cite{buhrman1998quantum,swaptest,gavinsky:exponential,yossef:exponential,gavinsky2019bare} and in this paper we show such separations can also be obtained in the memoryless models. We provide the proof in Theorems~\ref{thm:exponentiallyweakinequality} and~\ref{thm:qcseparations}.

In this paper, we haven't been able to give a large separation between $\QNM$ and $\NM$, primarily because all lower bound techniques we have for $\NM$ seem to apply for $\QNM$ as well such a deterministic one-way communication complexity and non-deterministic communication complexity.  The only ``trivial" separation we can give is a factor-$2$ gap between $\QNM$ and $\NM$ using the standard idea of super-dense coding (which in fact applies to \emph{all} classical memoryless protocols). Observe that by our our garden-hose characterization earlier, this factor-$2$ separation translates to quadratic separation between ``quantum-garden hose" model and the classical garden hose model.

Since $\NM(F)$ is at most  $\M^{\rightarrow}(F)$ for any $F$, memory-no memory communication complexity can be use to design garden-hose protocols. 
Using this method we obtain a sub-quadratic garden-hose protocol for computing the function \emph{Disjointness with quadratic universe} which was conjecture to have a quadratic complexity in~\cite{podder:garden}. We discuss this protocol in Section~\ref{sec:garden}.

\paragraph{5. Towards obtaining better formula bounds.} Finally, it was shown by Klauck and Podder~\cite{podder:garden} that any formulae of size $s$ consisting of \emph{arbitrary} fan-in 2 gates (i.e., formulae over the binary basis of fan-in 2 gates) can be simulated by a garden-hose protocol of size $s^{1+\varepsilon}$ for any arbitrary $\varepsilon > 0$. In this work, we show that an arbitrary garden-hose protocol can be simulated by a memoryless protocol without \emph{any} additional loss, i.e., a size $s$ garden-hose protocol can be turned into a memoryless protocol of size $\log s$. 
In particular, putting together these two connections, it implies that a size $s$ formula can be turned into a memoryless protocol of size $(1+\varepsilon)\log s$. 
Thus our result provides a new way of proving formulae size lower bound for arbitrary function $F$ by analyzing the memoryless protocol of $F$.\footnote{Here, the inputs $x,y$ are distributed among two players and their goal is to compute $(F\circ g)(x,y)$ where $g$ is a constant-sized gadget.} 
The best known lower bound for formulae size (over the basis of all fan-in 2 gate) is $\Omega(n^2/\log n)$, due to  Ne\v{c}iporuk from 1966~\cite{nechiporuk:cc}.
Analogous to the Karchmer-Wigderson games~\cite{karchmer:game} and Goldman and H\aa{}stad~\cite{goldmanhastad:cc} techniques which uses communication complexity framework to prove circuit lower bounds our new communication complexity framework is a new tool for proving formulae size lower bounds.

Brody et al.~\cite{brody2013space} conjectured that the problem of reachability in a graph requires $\Omega(\log^2 n)$ non-oblivious memory.
However as we have mentioned earlier the space bounded communication complexity and the memoryless communication complexity of any function are equal up to a constant factor. Thus proving this conjecture would imply the same lower bound on the memoryless communication complexity and in turn  imply an $n^{\log n}$ formula-size lower bound for reachability, which would be a break-through in complexity theory. In fact, because of the same general formula to memoryless communication simulation, showing even a $(2+\varepsilon) \log n$ lower bound for reachability would be very interesting.

Finally an additional benefit to our characterization is the following: information theory has been used extensively to understand communication complexity~\cite{bar2002information,bar2004information,ganor2015exponential,braverman2013information,braverman2015interactive} (just to cite a few references). As far as we are aware, using information theoretic techniques haven't been explored when understanding the models of computation such as formula size, branching programs and garden hose model. We believe our characterization using memoryless communication model might be an avenue to use information-theoretic ideas to prove stronger lower bounds in these areas.\footnote{In some sense, the message size in a protocol in the fully memoryless setting dictates how much information the players must store in order to pause and get back to the computation later. Note that as the players do not have any memory all such information has to be kept in the message register. Quantifying this amount of information for any function would yield a lower bound on the memoryless communication complexity of that function. In particular, note that in an $\NM$ protocol, suppose Alice and Bob exchanges the messages $\{\textbf{M}_i\}$, then for $i\neq j$, we have $\textsf{I}(\textbf{M}_i:\textbf{M}_j)=\textsf{H}(\textbf{M}_i)+\textsf{H}(\textbf{M}_j)$ (where $\textsf{I},\textsf{H}$ are mutual information and entropy of random variables respectively) since $M_i,M_{i+1}$ are related by a \emph{deterministic} function $f_x,g_y$ that are fixed once Alice and Bob obtain their inputs $x,y$, so $\textsf{H}(\textbf{M}_i| \textbf{M}_j)=0$.}

\subsection{Other related works.}
Finally here we discuss some more related works.  
Impagliazzo and Williams~\cite{impagliazzo2010communication} considered two variants of the standard communication complexity model where players have access to a synchronized clock. They then studied their relationship with the standard communication complexity and the polynomial hierarchy.
In the quantum setting, Ablayev et al.~\cite{ambainis:spacebounded} consider the memoryless communication model we consider and their focus was on discussing its connections to proving lower bounds on automata and ordered binary decision diagrams and streaming algorithms.   Chailloux et al.~\cite{chailloux2017information} study the quantum memoryless communication complexity wherein Alice and Bob do not have private memory but are allowed to adaptively apply unitaries to the quantum states they exchange in each round. They study
the information cost of memoryless quantum protocols and prove a tight lower bound on the information cost of the AND function for $k$-round quantum memoryless protocols.
Buhrman et al.~\cite{buhrman2016quantum} studied quantum memoryless protocols and established a connection between memoryless protocols and Bell inequality violations. 
Jeffery~\cite{jeffery:span} recently related the space complexity of quantum query algorithms with approximate span programs.  

\subsection{Open questions and future directions}
\label{sec:openfuture}
In this work, our main contribution has been to describe a seemingly-simple model of communication complexity  and characterize it's complexity using branching programs. We believe that our work could open up a new direction of research and results in this direction. Towards this, here we mention a few open questions:
\begin{enumerate}
    \item Is there a function $F:\01^n\times \01^n\rightarrow \01$ and a universal constant $c>1$ for which we have $\NM(F)\geq c \log n$. In particular, there are two consequences of such a result: (a) Using our relation to garden-hose model in Section~\ref{sec:garden}, such a function will lead to the first \emph{super-linear} $n^c$ lower bound for garden-hose complexity, (b) using our characterized to branching programs, this would result in the first super-linear $n^c$ lower bound for \emph{bipartite} branching programs (analogous to Tal's first super-linear lower bound on bipartite formula size of inner-product~\cite{tal:quadraticbipartite}). Also if we could show this for $c\geq 2+\varepsilon$, this would imply a $\Omega(n^{2+\varepsilon})$ lower bound for general formula size, improving upon the best lower bound by Ne\v{c}iporuk~\cite{nechiporuk:cc}.
   \item One possible candidate function which we haven't been to rule out is the distributed $3$-clique function: suppose Alice is given $x\in \01^{\binom{n}{2}}$ and Bob is given $y\in \01^{\binom{n}{2}}$. We view their inputs as jointly labelling of the $\binom{n}{2}$ edges of a graph on $n$ vertices, then does the graph with edges labelled by $x\oplus y$ have a triangle? Also, what is the complexity of the $k$-clique problem?
   \item Using information theory for proving better lower bounds? Our $\NM$ protocol by construction has several structural properties making it amenable to information-theoretic analysis. Could one use this to prove a super-linear garden-hose or bipartite branching program lower bound?
   
   \item Can we prove a $(2+\varepsilon)$-lower bound on 
   the $\NM$ complexity of reachability?
      \item In the memory-nomemory model of communication, we allowed Alice to have unlimited memory, but it is also interesting to consider the scenario where Alice's local memory is bounded (akin to the models considered by~\cite{memoryless:comm,brody2013space})?

    \item Another possible generic technique to approach this question is through a \emph{lifting theorem}: in particular, is there a gadget $g:\01^b\times \01^b\rightarrow \01$ such that for every function $F:\01^{nb}\times \01^{nb}\rightarrow \01$ and $F\circ g$ defined as 
    $$
    F\circ g(x^1,\ldots,x^n,y^1,\ldots,y^n)=F(g(x^1,y^1),\ldots,g(x^n,y^n)),
    $$ we have $\NM(F\circ g)\geq \NM \textsf{query}(F)\cdot b$?\footnote{In this model, an $m$-bit memory query algorithm works in the space of $\01^{\log n+1+m}$ bits. Suppose the goal is to compute $f:\01^n\rightarrow \01$ on input $x$ and the state of the algorithm is $(i,b,w)\in \01^{\log n+1+m}$, then one step of the query algorithm corresponds to the following: one query to $i$ and obtain $b\oplus x_i$, operations on the workspace memory $(w,b)$ in order to produce a new index $i'$ to query the next turn. The cost of such a protocol is the smallest workspace memory $m$ required to compute a function $F$ (on the hardest input $x$). It is not hard to see that the complexity in this query model corresponds to the size of the smallest \emph{branching program} computing $F$.}
   
    \item For every $k\geq 1$, can we find problems that are complete for the class of memoryless communication complexity when Alice and Bob share at most $k$ bits of memory?
    \item In this work, we only looked at the ``deterministic version'' of memoryless communication complexity. One could also look at the model wherein Alice and Bob have private or public randomness and need to compute $F$ with probability at least $2/3$.
 
    \item Is there a notion of \emph{catalytic communication complexity} (similar to the notion of catalytic computation introduced by Buhrman et al.~\cite{buhrman:catalytic}) wherein on top of memoryless-ness, the message register is filled with a pre-loaded data and at the end of the computation Alice and Bob needs to restore the data. It seems that in order to do this, the computation has to be reversible i.e., from Alice's (Bob's) perspective given an arbitrary input $x$ (input~$y$) two incoming messages \emph{cannot} be mapped into a single message. We remark that this model looks similar to the garden-hose model of computation which is also reversible and memoryless~\cite{podder:garden}.
\end{enumerate}

\paragraph{Organization.} 
In Section~\ref{sec:prelim} we describe the basic communication complexity model and branching programs. Section~\ref{sec:memcommmode} defines the memoryless communication model as well as other variants of this model. 
In Section~\ref{sec:charintermsofBP} we present various algorithms and characterize the complexity of memoryless communication in terms of bipartite branching program sizes. 
Finally, Section~\ref{sec:understandmem} presents some separations between the memoryless models and describes the equivalence between the memoryless model and the garden-hose and space bounded models.

\section{Preliminaries}
\label{sec:prelim}

 \paragraph{Notation.} Let $[n]=\{1,\ldots,n\}$. For $x\in \01^n$, let $\Int(x)\in \{0,\ldots,2^n-1\}$ be the integer representation of the $n$-bit string $x$. We now define a few standard functions which we use often in this paper. The equality function $\EQ_n:\01^n\rightarrow \01^n\rightarrow \01$ is defined as $\EQ_n(x,y)=1$ if and only if $x=y$. The disjointness function $\DISJ_n$ defined as $\DISJ_n(x,y)=0$ if and only if there exists $i$ such that $x_i=y_i=1$. The inner product function $\IP_n$ is defined as $\IP(x,y)=\sum_{i} x_i\cdot y_i \pmod{2}$ (where $\cdot$ is the standard bit-wise product).

 \paragraph{Quantum information.} We briefly review the basic concepts in quantum information theory. Here a qubit $\ket{\psi}$ is a unit vector in $\mathbb{C}^2$ and the basis for this space is denoted by $\ket{0}=\left(\begin{array}{c}1 \\ 0 \end{array}\right)$ and $\ket{1}=\left(\begin{array}{c} 0 \\ 1 \end{array}\right)$. An arbitrary $\ket{\psi}$ is in a \emph{superposition} of $\ket{0},\ket{1}$, i.e., $\ket{\psi}=\alpha\ket{0}+\beta\ket{1}$ for $\alpha,\beta\in \mathbb{C}$ satisfying $|\alpha|^2+|\beta|^2=1$.  In order to obtain a quantum state on $n$ qubits, one can take the tensor product of single-qubit states, hence an arbitrary $n$-qubit state $\ket{\phi}$ is a unit vector in $\mathbb{C}^{2^n}$  and can be expressed as $\ket{\psi}=\sum_{x\in \01^n}\alpha_x \ket{x}$ where $\alpha_x\in \mathbb{C}$ and $\sum_x |\alpha_x|^2=1$. 

We now define formulae, branching programs and refer the interested reader to Wegener's book~\cite{wegener1987complexity} for more on the subject. 

\begin{definition}[De Morgan Formulae]
A De Morgan formula is a binary
tree whose internal nodes are marked with AND gates or OR gates and the leaves are marked with input variables $x_1, x_2, \cdots , x_n$ or their negations.
we say a depth-$d$ De Morgan formula computes a function $f:\01^n\rightarrow \01$, if there is a binary tree of depth $d$ such that on input $x= (x_1, x_2, \cdots , x_n)$, the root of the tree outputs $f(x)$. The De Morgan formula size of a function $f$ is the size of the smallest De Morgan formula computing $f$. 
\end{definition}

\begin{definition}[General Formulae]
General formulas are same as De Morgan formulas except that the nodes could consist of gates corresponding to any arbitrary $2$-bit functions $f:\01^2 \rightarrow \01$. Note that De Morgan formulas consisted only of AND, OR gates whereas in general formulas there could be 16 such possible gates in the tree.
\end{definition}

\vspace{0mm}
	\begin{definition}[Branching programs $(\BP)$]
		A branching program for computing a Boolean function $f:\01^n\rightarrow \01$ is a directed acyclic graph with a source node labelled  $S$ and two sink nodes labelled $0$ and $1$. Every node except the source and sink nodes are labelled by an input variable $x_i$. The out-degree of every node is two and the edges are labelled by $0$ and $1$. The source node has in-degree $0$ and the sink nodes have out-degree $0$. The size of a branching program is the number of nodes in it. We say a branching program computes $f$ if for all $x\in f^{-1}(1)$ (resp.~$x\in f^{-1}(0)$) the algorithm starts from the source, and depending on the value of $x_i\in \01$ at each node the algorithm either moves left or right and eventually reaches the $1$-sink (resp.~$0$-sink) node. We denote $BP(f)$ as the size (i.e., the number of nodes) of the smallest branching program that computes $f$ for all $x\in \01^n$.
		\end{definition}

	We now define the standard communication complexity model defined by Yao~\cite{yao:comm}.
	
\vspace{0mm}
\begin{definition}[Standard communication complexity]
\label{def:ccyao}
Let $F:\01^n\times \01^n\rightarrow \01$. Here 
two players Alice and Bob want to compute $F$ in the following manner: Alice receives $x\in \01^n$, Bob gets $y\in \01^n$ and they are allowed to exchange bits before computing $F(x,y)$. We say a protocol computes $F$ if for every $x,y$, Alice and Bob compute $F(x,y)$ with probability $1$.  
	The models we describe below vary in how the communication protocol proceeds between Alice and Bob and measures the complexity of the protocols in a different way. 
	\begin{enumerate}
		\item Standard \emph{one-way} communication complexity: Here we restrict only Alice to send bits to Bob. The complexity of the protocol in this model is the total number of communication between Alice and Bob and finally Bob needs to output $F(x,y)$. The classical complexity of this model is denoted $D^{\rightarrow}(F)$. Suppose they exchange quantum bits, then due to the inherent randomness in quantum states, we allow them to output $F(x,y)$ with probability at least $2/3$. The quantum complexity of computing $F$ in this model is denoted by $Q^{\rightarrow}(F)$.
		
		\item  Standard \emph{two-way} communication complexity: This is exactly the same as $D^\rightarrow(F)$, except that both Alice and Bob are allowed to send bits to one another and the classical complexity is denoted $\textsf{D}(F)$. Suppose they are allowed to exchange quantum bits, then the quantum complexity is denoted  $\textsf{Q}(F)$
	\end{enumerate}
\end{definition}

\section{Memoryless Communication Complexity}
\label{sec:memcommmode}
	In this section we define memoryless communication complexity model and its variants. 

	\subsection{Deterministic Memoryless Communication Model}
	The crucial difference between the memoryless communication model and  standard communication model is that, at any round of the communication protocol Alice and Bob do not have memory to remember previous transcripts and their private computations from the previous rounds. We now make this formal. 

\vspace{0mm}
	\begin{definition}[Two-way Deterministic memoryless communication complexity]
		\label{definition:twowaymemoryless}
         Let $F:\01^n\times \01^n\rightarrow \01$. Here there are two parties Alice and Bob whose goal is to compute $F$. Every $s$-bit memoryless protocol is defined by a set of functions $\{f_x\}_{x\in \01^n}$ and $\{g_y\}_{y\in \01^n}$ wherein $f_x,g_y:\01^s\rightarrow \01^s$. On input $x,y$ to Alice and Bob respectively a memoryless protocol is defined as follows: 
 		 at every round Alice obtains a message $m_B\in \01^s$ from Bob, she computes $m_A=f_x(m_B)\in \01^s$ and sends $m_A$ to Bob. On receiving $m_A$, Bob computes $m_B'=g_y(m_A)$ and replies with $m'_B\in \01^s$ to Alice. They alternately continue doing this for every round until the protocol ends. Without loss of generality we assume the protocol ends once $m_A,m_B\in \{1^{s-1}0,1^{s-1}1\}$, then the function output is given by the last bit. So, once the transcript is $1^{s-1}b$, Alice and Bob output $F(x,y)=b$.\footnote{Without loss of generality, we assume that the first message is between Alice and Bob and she sends $f_x(0^s)\in \01^s$ to Bob.}

		We say a protocol \emph{$P_F$ computes $F$ correctly} if for every $(x,y)$,  Bob outputs $F(x,y)$. We let $\textsf{cost}(P_F,x, y)$ be the smallest $s$ for which $P_F$ computes $F$ on input $(x,y)$. Additionally, we~let
        $$
        \textsf{cost}(P_F) = \max_{x,y} \textsf{cost}(P_F,x, y)
        $$
	and the \emph{memoryless communication complexity} of computing $F$ in this model is defined as
		$$
		\NM(F) = \min_{P_F} cost(P_F),
		$$
		where is the minimum is taken over all protocols $P_F$ that compute $F$ correctly.
		\end{definition}
	We crucially remark that in the memoryless model, the players do not even have access to a clock and hence they cannot tell which round of the protocol they are in. At every round they just compute their local functions $\{f_x\}_x,\{g_y\}_y$ on the message they received and proceed according to the output of these  functions.

	\paragraph{One-way Deterministic Memoryless Model.} 
	Similar to the definition above, one can define the \emph{one-way} memoryless communication complexity wherein only Alice is allowed to send messages to Bob and the remaining aspects of this model is the same as Definition~\ref{definition:twowaymemoryless}. We denote the complexity in this model by $\NM^\rightarrow(F)$.
	It is easy to see that since Alice does not have any memory she cannot send multi-round messages to Bob as there is no way for her to remember in which round she is in. Also Bob cannot send messages back to Alice for her to keep a clock. Hence all the information from Alice to Bob has to be conveyed in a single round. Thus one-way memoryless communication complexity is equal to the standard deterministic one-way communication~complexity.\footnote{Without loss of generality, in any one-way standard communication complexity protocol of cost $c$ Alice can send all the $c$ bits in a single round.}

\vspace{0mm}
\begin{fact}
    \label{fact:nm-oneway_equal_det_one_way}
    For all function $F$ we have $\NM^\rightarrow(F) = \textsf{D}^\rightarrow(F)$. 
    \end{fact}

    \subsection{Deterministic Memory-No Memory Communication Model}
We now consider another variant of the memoryless communication model wherein one party is allowed to have a memory but the other party doesn't. 
In this paper, we always assume that Alice has a memory and call this setup the \emph{memory no-memory model}. In this work, we will \emph{not} consider the other case wherein Bob has a memory and Alice doesn't have a memory. Note that this setting is asymmetric i.e., there exists functions for which the complexity of the function can differ based on whether Alice or Bob has the memory.

\paragraph{Two-way Memory-No Memory Communication Model.} 
Here the players are allowed to send messages in both directions. For a function $F:\01^n\times \01^n\rightarrow \01$, we denote the complexity in this model as $\M(F)$. Observe that $\M(F)$ is trivially upper bounded by $\log n$ for every~$F$: for every $i\in [n]$, Alice can send $i$ and Bob replies with $y_i$. Since Alice has memory, after~$n$ rounds she has complete knowledge of $y\in \01^n$ and computes $F(x,y)$ locally and sends it to Bob.

    \vspace{0mm}
\paragraph{One-way Memory-No Memory Communication Model.}
Here we allow only Alice to send messages to Bob. Since Alice has a memory she can send multiple messages one after another, but Bob cannot reply to her messages. Hence, after receiving any message Bob computes the function $g_y(\cdot)\in \{0,1,\perp\}$ and if he obtains $\01$, he outputs $0$ or $1$, and continues if he obtains~$\perp$. 
We denote the communication complexity in this model by $\M^{\rightarrow}(F)$. 
This model was formally studied by Papakonstantinou et al.~\cite{memoryless:comm} as \emph{overlay communication complexity} (we discuss their main contributions in Section~\ref{sec:charintermsofBP}). 

Finally, we can also have a model where both players have memory and hence both players can remember the whole transcript of the computation. This is exactly the widely-studied standard communication complexity except that the complexity measure here is the size of the \emph{largest} transcript (so the complexity in our model is just $1$ since they could exchange a single bit for $n$ rounds and compute  an arbitrary function on $2n$ bits) and the latter counts the \emph{total} number of bits exchanged in a protocol.

\vspace{0mm}
\paragraph{Quantum memoryless Models.} Here we introduce the quantum memoryless communication model. There are a few ways one can define the quantum extension of the classical memoryless model. We find the following exposition the simplest to explain. This quantum communication model is defined exactly as the classical memoryless model except that Alice and Bob are allowed to communicate \emph{quantum} states. A $T$~round quantum protocol consists of the following: Alice and Bob have local $k$-qubit memories $\mathsf{A},\mathsf{B}$ respectively,\footnote{After each round of communication, these registers  are set to the all-$0$ register.} they share a $m$-qubit message register~$\mathsf{M}$ and for every round they perform a $q$-outcome POVM $\mathcal{P}=\{P_1,\ldots,P_{q}\}$ for $q=2^m$ (which could potentially depend on their respective inputs $x$ and $y$). Let $\{U_x\}_{x\in \01^n},\{V_y\}_{y\in \01^n}$ be the set of $(m+k)$-dimensional unitaries acting on $(\mathsf{A},\mathsf{M})$  and $(\mathsf{B},\mathsf{M})$ respectively (this is analogous to the look-up tables $\{f_x,g_y:\01^m\rightarrow \01^m\}_{x,y\in \01^n}$  used by Alice and Bob in the classical memoryless protocol). Let $\psi_0=(\mathsf{A},\mathsf{M})$ be the all-$0$ mixed state. Then, the quantum protocol between Alice and Bob can be written as follows: on input $x,y$ to Alice and Bob respectively, on the $i$th round (for $i\geq 1$) Alice sends ${\psi_i}$ for odd $i$ and Bob replies with ${\psi_{i+1}}$ defined as follows: 
$$
\psi_i=\Tr_{\mathsf{A}} \big(\mathcal{P}\circ U_x {\psi_{i-1}}\big)\otimes \ketbra{0}{0}_{\mathsf{B}},
$$
where $\mathcal{P}\circ U_x {\psi_{i-1}}$ is the post-measurement state after performing the POVM $\mathcal{P}$ on the state $U_x \psi_{i-1}$ and $\Tr_A(\cdot)$ refers to taking the partial trace of register $\mathsf{A}$.  Similarly, define 
$$
\psi_i=\ketbra{0}{0}_{\mathsf{A}} \otimes \Tr_{\mathsf{B}} \big(\mathcal{P}\circ U_y {\psi_{i}}\big),
$$
where $\Tr_{\mathsf{B}}(\cdot)$ takes the partial trace of  register $\mathsf{B}$. Intuitively, the states $\psi_i$ (similarly $\psi_{i+1}$) can be thought of as follows: after applying unitaries $U_x$ to the registers $(\mathsf{A},\mathsf{M})$, Alice applies the $q$-outcome POVM $\mathcal{P}$ which results in a classical outcome and post-measurement state on the registers $(\mathsf{A},\mathsf{M})$ and she discards her private memory register and initializes the register $\mathsf{B}$ in the all-$0$ state. The quantum communication protocol terminates at the $i$th round once the $q$-outcome POVM $\mathcal{P}$ results in the classical outcome $\{(1^{m-1},b)\}_{b\in \01}$.\footnote{We remark that a good quantum communication protocol should be such that for every $i\in [T]$, the probability of obtaining $(1^{m-1},1\oplus F(x,y))$ when measuring $\psi_i$ using the POVM $\mathcal{P}$ should be $\leq 1/3$.} After they obtain this classical output, Alice and Bob output $b$.
We say a \emph{protocol computes $F$} if for every $x,y\in \01^n$, with probability at least $2/3$~(probability taken over the randomness in the protocol), after a certain number of rounds the POVM measurement results in $(1^{m-1},F(x,y))$. The complexity of computing $F$ in the quantum memoryless model, denoted $\QNM(F)$ is the smallest $m$ such that there is a $m$-qubit message protocol that computes $F$. As defined before, we also let $\QM^{\rightarrow}(F)$ (resp.~$\QNM^{\rightarrow}(F)$) to be the model in which Alice has a memory (has no memory) and Bob doesn't have a memory and the communication happens from Alice to Bob.

\paragraph{Notation.} For the remaining part of the paper we abuse notation by letting $\NM(F)$, $\QNM(F)$ denote the memoryless \emph{complexity} of computing $F$ and we let  \emph{$\NM$ model} (resp.~\emph{$\QNM$ model}) be the  memoryless communication model (resp.~quantum memoryless communication model). Additionally, we omit \emph{explicitly} writing that Alice and Bob exchange the final message $1^{s-1}f(x,y)$ once either parties have computed $f(x,y)$ (on input $x,y$ respectively).

\section{Understanding and characterization of memoryless models}
\label{sec:charintermsofBP}
We now state a few  observations and relations regarding the memoryless communication~models.

\vspace{0mm}
	\begin{fact}
		\label{claim:firstinequality}
		For every $F:\01^n\times \01^n\rightarrow \01$, we have
		$$
		\M(F)\leq 	\NM(F)  \leq	2\M^{\rightarrow}(F)\leq	2\NM^{\rightarrow}(F).
		$$
	\end{fact}
\begin{proof}
The proof of the first and last inequality is straightforward since the complexity of a protocol only increases when we force a party to not have memory. We now show the second inequality, suppose $\M^{\rightarrow}(F)=k$. Then there are at most $2^k$ messages Alice sends (suppose at round $i$, she sends a message $m$ and Bob didn't output $\{0,1\}$, then she knows that $g_y(m)=\perp$, so she need not repeat sending $m$ using her memory) and suppose these $2^k$ messages Alice sends are $\{f_x(m_1),\ldots,f_x(m_{2^k})\}$. The $\NM$ protocol for $f$ goes as follows, at the $i$th round Alice sends Bob $(f_x(m_i),i)$ which uses $2k$ bits. If Bob does not outout $\{0,1\}$ he simply increments $i$ to $i+1$ and sends back $i+1$, which takes at most $k$ bits. Since Alice doesn't have a memory but has received $i+1$ and hence has the information of the round $i+1$, she runs the $\M^\rightarrow(F)$ protocol for $(i+1)$th round and her next message to Bob is $(f_{x}(m_{i+1}),i+1)$. Alice and Bob continue the protocol this way until Bob outputs 0 or 1. Since the original $\M^{\rightarrow}(F)$ protocol computes $F$, the $\NM$ protocol computes $F$ as well. 
\end{proof}

	 As we mentioned earlier, our main contribution in this paper is the  memoryless $\NM$ model of communication. We saw in Fact~\ref{fact:nm-oneway_equal_det_one_way} that $\NM^{\rightarrow}(F)$ is equal to the standard one-way deterministic communication complexity of computing $F$. The $\M^{\rightarrow}(F)$ model was introduced and studied by Papakonstantinou et al.~\cite{memoryless:comm}. Additionally observe that the strongest model of communication complexity~$\M(F)$ is small for every function $F$.

\vspace{0mm}
\begin{fact}
\label{fact:MQMissmall}
    For every $F:\01^n\times \01^n\rightarrow \01$, we have $\M(F)\leq \log n$.
    %and $\QM(F)=1$.
\end{fact}
	To see this, observe that in the $\M$ model (i.e., two-way memory-no memory model), on the $i$th round, Alice sends $i\in [n]$ and Bob (who doesn't have memory) sends the message $y_i$ to Alice. Alice stores $y_i$ and increments $i$ to $i+1$ and repeats. After $n$ rounds Alice simply has the entire $y$ and computes $F(x,y)$ on her own (note that $F$ is known to both Alice and Bob). 
	
Below we give few protocols in the $\NM$ model to give more intuition of this~model. 

\paragraph{Algorithms in the memoryless model:}
In the introduction we described a $\log n+1$ protocol for the equality function. Below we describe a protocol for the inner product function.
For the inner product function $\IP_n$, a simple protocol is as follows: For $i=1,\ldots,n$, on the $i$th round, Alice~sends  
		$$
		\Big(i,x_i, \sum_{j=0}^{i-1}x_i \cdot y_i\pmod 2  \Big)
		$$ 
		which takes $\log n+2$ bits and Bob replies with 
		$$
		\big(i,x_i, \sum_{j=0}^{i-1}x_i \cdot y_i+x_i\cdot y_i \pmod 2 \Big)=\Big(i,x_i, \sum_{j=0}^{i}x_i \cdot y_i \pmod 2 \Big).\footnote{Technically Bob need not send back the bit $x_i$.}
		$$
		They repeat this protocol for $n$ rounds and after the $n$th round, they have computed $\IP_n(x,y)$. Hence $\NM(\IP_n)\leq \log n+2$. Now we describe a protocol for the disjointness function $\DISJ_n$. Here a $\log n$ protocol is as follows: Alice sends the first coordinate $i\in [n]$ for which $x_i=1$ and Bob outputs~$0$ if $y_i=1$, if not Bob replies with the first $j$ after $i$ for which $y_j=1$ and  they repeat this procedure until $i$ or $j$ equals $n$. It is not hard to see that  $\DISJ_n(x,y)=0$ if and only if there exists $k$ for which $x_k=y_k=1$ in which case Alice and Bob will find such (smallest) $k$ in the protocol above, if not the protocol will run for at most $n$ rounds and they decide that $\DISJ_n(x,y)=1$. 
We now mention a non-trivial protocol in the $\NM$ model for the majority function defined as $\MAJ_n(x,y)=\big[\sum_i x_i\cdot y_i \geq n/2+1]$. A trivial protocol for $\MAJ_n$ is similar to the $\IP_n$ protocol, on the $(i+1)$th round, Alice sends $(i,x_i,\sum_{i=1}x_i\cdot y_i)$ (without the $\pmod 2$) and Bob replies with $(i,x_i,\sum_{i=1}^{n+1}x_i\cdot y_i)$. Note that this protocol takes $2\log n+1$ bits ($\log n$ for sending the index $i\in [n]$ and $\log n$ to store $\sum_{i=1}^{n}x_i\cdot y_i \in [n]$). Apriori this seems the best one can do, but interestingly using intricate ideas from number theory there exists a $n\log^3 n$~\cite{sinha:majorityBP,podder:garden} garden-hose protocol for computing $\MAJ_n$. Plugging this in with Theorem~\ref{lem:nmghm} we get a protocol of cost $\log n+3\log \log n$ for computing $\MAJ_n$ in the $\NM$ model.

An interesting question is, are these protocols for $\IP_n, \EQ_n$, $\DISJ_n$, $\MAJ_n$ optimal? 
 Are there more efficient protocols possibly with constant bits of communication in each round? In order to understand this, in the next section we show that the memoryless communication complexity is lower bounded by the standard detereministic one-way communication complexity. Using this connection, we can show the tightness of the first three protocols. Additionally, we show that $\NM(\MAJ_n)\geq \log n$, thus the exact status of $\NM(\MAJ_n)\in \{\log n,\ldots,\log n+3\log \log n\}$  remains an intriguing open~question.

 \subsection{Lower bounds on memoryless communication complexity}
 In the introduction, we mentioned that it is an interesting open question to find an explicit function~$F$ for which $\NM(F)\geq 2\log n$. Unfortunately we do not even know of an explicit function for which we can prove lower bounds better than $\log n+\omega(1)$ (we discuss more about this in the open questions). However, it is not hard to show that for a random function $F$, the memoryless communication complexity of $F$ is large.

	\begin{lemma}
		\label{lemma:randomfhasighR}
		Let $F:\01^n\times \01^n\rightarrow \01$ be a random function. Then,
		$
		\NM(F)=\Omega(n).
		$
	\end{lemma}
	
\begin{proof}
The proof is via a simple counting argument. There are $2^{2^{2n}}$ distinct functions $F:\01^n\times \01^n\rightarrow \01$.  Consider an arbitrary $s$-bit $\NM$ protocol. Let Alice and Bob's local functions be given by $\{f_x:\01^s\rightarrow \01^s\}_{x\in \01^n}$ and similarly $\{g_y\}_y$. 
	First observe that there are at most $2^s$ distinct messages that Alice can receive from Bob and there are at most $2^n$ distinct inputs $x$, so the total number distinct messages that Alice can send to Bob (recall that Alice's messages are given by $f_x(\cdot)\in \01^s$) is at most $(2^s)^{2^s\cdot 2^n} = (2^s)^{2^{s+n}}$. Similarly Bob can send at most $(2^s)^{2^{s+n}}$ distinct messages to Alice. In total there are at most $(2^s)^{2\cdot 2^{s+n}}= (2^s)^{\cdot 2^{s+n+1}}$ distinct protocols that can arise from an $s$-bit $\NM$ protocol. If we have an $\NM$ protocol that computes an arbitrarily random function $F:\01^n\times \01^n\rightarrow \01$, then we need that $(2^s)^{\cdot 2^{s+n+1}} \geq 2^{2^{2n}}$, which implies $s \geq n -\log n -1$.	
\end{proof}	

We remark that similar ideas used in this lemma can be used to show that for all $s<s'$, there exists functions that can be computed using $s'$ bits of communication in each round but not $s$ bits of communication. This gives rise to a \emph{space hierarchy} theorem for the $\NM$ model.

\subsubsection{Deterministic one-way communication complexity and memoryless complexity.}
We now give a very simple lower bound technique for the memoryless communication model in terms of deterministic one-way communication. Although this lower bound is ``almost immediate", as we mentioned in the introduction, it already gives us non-trivial lower bounds on the $\NM$ complexity of certain functions.
    \begin{fact}
    \label{lem:nondeternm}
    $\NM(F) \geq \log \Big(\textsf{D}^{\rightarrow}(F)/ \log \textsf{D}^{\rightarrow}(F)\Big)$.
    \end{fact}
    \begin{proof}
    The proof is trivial. Let $\NM(F)=t$. Once Alice and Bob receive $x,y$ respectively, by definition of the $\NM$ model, Alice has a function $f_x:\01^t\rightarrow \01^t$. She simply sends the truth table of $f_x$ to Bob, which consists of $t\cdot 2^t$ bits. Once Bob receives $f_x$, he can simply simulate the $\NM(F)$ protocol on his own. This gives us our fact statement.
    \end{proof}

\begin{fact}
    $\QNM(F) \geq \Omega \Big(\log \Big(\textsf{D}^{\rightarrow}(F)/ \log \textsf{D}^{\rightarrow}(F)\Big)\Big)$
\end{fact}

The proof of this fact is exactly the same as Fact~\ref{lem:nondeternm}. Here, instead of Alice sending the truth table $f_x:\01^t\rightarrow \01^t$, Alice instead sends the $t$-qubit unitary $U_x$ which is a matrix in $\mathbb{R}^{2^t\times 2^t}$. So this involves Alice sending $2^{2t}\times \poly(t)$ many bits to Bob (the $\poly(t)$-factor is to ensure that Alice sends each entry up to $2^{-t}$ precision). Bob can then locally simulate the entire quantum $\NM$ protocol (without requiring any further information from Alice) and thereby computes $F$. 

    Using this lemma, we immediately get the following corollary.
    \begin{corollary}
    \label{cor:nmislogn}
Let $n\geq 2$. Then $\NM(\EQ_n),\NM(\IP_n),\NM(\DISJ_n),\NM(\MAJ_n),\NM(\mathsf{Index}),\NM(\mathsf{BHM})$ is $\Omega(\log n)$. Similarly, we have  $\QNM$ complexity of these functions are $\Omega (\log n)$.
\end{corollary}
This corollary follows immediately from Fact~\ref{lem:nondeternm} because the detereministic-one way communication complexity of these functions are at least $n$ (by a simple adverserial argument), thereby showing that the $(\log n)$-bit protocols we described in the beginning of this section for the first three of these functions is close-to-optimal. However one drawback of Fact~\ref{lem:nondeternm} is it cannot be used to  prove a lower bound that is better than $\log n$ since $\textsf{D}^{\rightarrow}(F)\leq n$ for every function $F:\01^n\times \01^n\rightarrow \01$.
    
    \subsection{Characterization of memoryless communication}\label{subsec:overlay}
    
Papakonstantinou et al.~~\cite{memoryless:comm} consider the memory-nomemory model of communication complexity wherein Alice has a memory and Bob doesn't and they are restricted to one-way communication from Alice to Bob. They show a beautiful combinatorial rectangle-overlay characterization (denoted $\textsf{RO}(F)$) of the  $\M^{\rightarrow}$ model. We briefly define $\textsf{RO}(F)$ below.

	\begin{definition}[Rectangle overlay complexity \cite{memoryless:comm}]
	\label{def:overlay}
		Let $n\geq 1$, $F:\01^n\times \01^n\rightarrow \01$. A \emph{length-$k$ rectangle overlay} for $F$ is a collection $\{(R_1,b_1),\ldots,(R_k,b_k) \}$ satisfying the following:
		
		\begin{itemize}
			\item $R_i\subseteq \01^n\times \01^n$ is a  $b_i$-monochromatic rectangle (i.e., $R_i=X_i\times Y_i$ where $X_i,Y_i\subseteq \01^n$) and $b_i\in \01$. 
			\item $\{R_1.\ldots,R_k\}$ covers $\01^n\times \01^n$.
			\item For every $(x,y)\in \01^n\times \01^n$, suppose $R_\ell$ was the \emph{first} rectangle that contains $(x,y)$, then $F(x,y)=b_\ell$.
		\end{itemize} 
		Then, define $\mathsf{RO}(F)$ is the smallest $k$ for which there exists a length-$k$ rectangle overlay for~$F$.
	\end{definition}
	One of the main results of~\cite{memoryless:comm} was the following characterization.
	
	\begin{theorem}[\cite{memoryless:comm}]
		For every $F:\01^n\times \01^n\rightarrow \01$, we have
		$$
		\log \mathsf{RO}(F)\leq \M^{\rightarrow}(F)\leq 2\log \mathsf{RO}(F).
		$$
	\end{theorem}
       A natural question following their work is, can we even characterize our new general framework of communication complexity wherein \emph{both} Alice and Bob do not have memory and the communication can be two-way. Generalizing the rectangle-based characterization of~\cite{memoryless:comm} to our setting seemed non-trivial because in our communication model the memoryless-ness of the protocol doesn't seem to provide any meaningful way to split the communication matrix into partitions or overlays (as far as we could analyze). 
    Instead we characterize our communication model in terms of \emph{bipartite branching programs}, which we define below.\footnote{For a definition of general branching program ($\BP$), refer to Section~\ref{sec:prelim}.}
 
	\begin{definition}[Bipartite Branching Program ($\BBP$)]
	Let $F:	\01^n\times \01^n\rightarrow \01$. A bipartite branching program is a $\BP$ that computes $F$ in the following way: for every $(x,y)$, each node in the branching program is \emph{either} labelled by a function $f_i \in \mathcal{F}=\{f_i:\01^n\rightarrow \01\}_{i}$ \emph{or} by $g_j \in \mathcal{G}=\{g_j:\01^n\rightarrow \01\}_{j}$; the output edge is labelled by $0$ or $1$ and the output of the function in the node label decides which edge to follow. The size of a $\BBP$ is the number of nodes in it. We define	$\BBP(F)$ as the size of the smallest program that computes $F$ for all $(x,y)\in \01^{2n}$.
	\end{definition}
	
    Observe that in a $\BBP$ every node no longer just queries $x\in \01^n$ at an arbitrary index $i$ (like in the standard $\BP$), but instead is allowed to compute an \emph{arbitrary Boolean function} on~$x$ or~$y$. Of course, another natural generalization of $\BBP$ is, why should the nodes of the program just compute Boolean-valued functions? We now define the \emph{generalized} $\BBP$ wherein each node can have out-degree $k$ (instead of out-degree $2$ in the case of $\BBP$ and $\BP$).

	\begin{definition}[Generalized Bipartite Branching Program $(\GBBP)$]
	Let $k\geq 1$. A generalized bipartite branching program is a $\BBP$ that computes $F$ in the following way: for every $(x,y)$, each node in the branching program can have out-degree $k$ and labelled by the node $f_i \in \mathcal{F}=\{f_i:\01^n\rightarrow [k]\}_{i}$, or by $g_j \in \mathcal{G}=\{g_j:\01^n\rightarrow [k]\}_{j}$;  the output edges are labelled by $\{1,\ldots,k\}$ and the output of the function in the node label decides which edge to follow. The size of a $\GBBP$ is the number of nodes in it. We define	$\GBBP(F)$ as the size of the smallest program that computes $F$ for all $(x,y)\in \01^{2n}$.
		\end{definition}

We now show that the generalized bipartite branching programs are not much more powerful than bipartite branching programs, in fact these complexity measures are quadratically~related.
\begin{fact}
\label{fact:GBBPBBP}
For $F:\01^n\times \01^n\rightarrow \01$, we have $\GBBP(F)\leq \BBP(F)\leq \GBBP(F)^2$.
\end{fact}

\begin{proof}
The first inequality is obvious as $\GBBP$s are generalized version of $\BBP$s and thus can simulate $\BBP$s. Let $\GBBP(F)=s$. In order to see the second inequality, we show that every node in the $\GBBP$ can be computed using a $\BBP$ with at most $s$ nodes: observe that if a node in $\GBBP$ has $k$ outputs then we can express this node using a binary tree of depth $\log k$ and size $k$ such that each node in this binary tree is indexed by a Boolean function. Hence we can replace every node in the $\GBBP$ using the argument above, and we get $\BBP(F)\leq s\cdot k\leq s^2$.   
\end{proof}

It is not clear if the quadratic factor loss in the simulation above is necessary and we leave it as an open question. We are now ready to prove our main theorem relating $\NM$ communication model and bipartite branching programs.

	\begin{theorem}
	\label{thm:d_nomem_logT}
		For every $F:\01^{n\times n}\rightarrow \01$, we have
		$		\frac{1}{2}\log \BBP(F)\leq \NM(F)\leq \log \BBP(F)$.
	\end{theorem}
	\begin{proof}
    We in fact prove something stronger here, i.e., $\NM(F)=\log \GBBP(F)$ for all $F$. Using Fact~\ref{fact:GBBPBBP} we get the theorem statement.
    
    We first prove $\log \GBBP(F)\leq \NM(F)$. Let $\NM(F)=s$. Given an $s$-bit $\NM$ protocol that computes~$F$, one can label all possible messages from Alice to Bob by the set $\Mset_A=\{m^A_1, m^A_2, \ldots, m^A_{2^s}\}$ and similarly all the messages from Bob to Alice by $\Mset_B=\{m^B_1, m^B_2, \cdots m^B_{2^s}\}$.\footnote{Technically the number of messages exchanged between them should be a parameter $T$, but for notational simplicity we assume they communication for $2^s$ rounds, where we used $T\leq 2^s$ by the argument in the first paragraph of the proof of Fact~\ref{lem:nondeternm}.} Also let us suppose in the $\NM$ protocol, Alice and Bob have functions $\{f_x\}_{x\in \01^n}$ and $\{g_y\}_{y\in \01^n}$ respectively. Now we construct a generalized bipartite branching program that contains $2^{s+1}$ nodes, each node labelled by one of the messages $\Mset_A\cup \Mset_B$ which were exchanged between Alice and Bob in the~$\NM$ protocol. It remains to establish the edge-connections between the respective nodes as well as associate each node with a functions $\{f'_i\}_{i=1}^{2^s},\{g'_j\}_{j=1}^{2^s}$ each mapping $\01^s$ to  $\01^s$.  
    %, along with  two additional nodes labelled by $0^s$ and $1^s$
	For any two pair of nodes $(m^A,m^B)\in \Mset_A\times \Mset_B$ we put an edge from $m^A$ to $m^B$ \emph{if and only if} there exists some $y\in \01^n$ such that $g_y(m^A)=m^B$ in the $\NM$ protocol computing $F$. Similarly for any two pair of nodes $(m^B,m^A)\in \Mset_B\times \Mset_A$ we put an edge from $m^B$ to $m^A$ \emph{if and only if} there exists some $x\in \01^n$ such that $f_x(m^B)=m^A$ in the $\NM(F)$ protocol. We now associate each node with a function $f'_i,g'_i$.  On every node $m^B_j \in \Mset_B$, we associate the function $f_j': \01^n \rightarrow \01^s$ defined as $f_j'(x)\triangleq f_x(m^B_j)$. Similarly on every node $m^A_i \in \Mset_A$, we associate a function $g_i': \01^n \rightarrow \01^s$ defined as $g_i'(y)\triangleq g_y(m^A_i)$. These functions dictate how the branching program routes between every two nodes. Finally, we make $m^A_1$ as the source node and we glue the nodes $m^A_{1^{s-1}b}$ and $m^B_{1^{s-1}b}$ together and designate it as a sink node $b$. This completes the construction of the branching program. Note that all the edges are either from $\Mset_A$ nodes to $\Mset_B$ nodes or vice versa {but not both}. It now follows from the construction that for every $(x,y)$ the sequence of message (from $\Mset_A \cup \Mset_B$) exchanged between Alice and Bob is exactly the sequence of nodes traversed in the branching program when the functions $f'_i,g'_j$ are evaluated on the inputs $x,y$. By the promise of the $\NM$ protocol, we have that the branching program computes~$F(x,y)$, hence $\GBBP(F)\leq 2^s$.

    We now show the other direction, $\log\GBBP(F)\geq \NM(F)$.  Consider an arbitrary generalized bipartite branching program of size $s$. Note that one can view a $\GBBP$ computing $F$ as a bipartite graph. Let us assume that $s= k + \ell$, where the bi-partition sizes are $k$ and $\ell$. We now construct an~$\NM$ protocol of cost at most $\log s$. Suppose the nodes are labelled by $[k] \cup [\ell]$. We now let $a_i$ (resp.~$b_j$) be the binary number representations of the node label $[k]$ (resp.~$[\ell]$). By the definition of a $\GBBP$, there are no edges within the nodes labelled by $[k]$ and within nodes labelled by $[\ell]$ and we only have edges going between $[k]$ and $[\ell]$. 
    %Also the node $0$ and $1$ has zero out-degree. 
    Moreover, by definition every node labelled by $[k]$ (resp.~$[\ell]$) has a function $f_i$ (resp.~$g_j$) associated with it, where $f_i,g_j:\01^n \rightarrow \01^s$. We now construct our $\NM$ protocol as follows:  we define two sets of function $\{f_x'\}_{x\in \01^n}$ and $\{g_y'\}_{y\in \01^n}$ where $f_x',g_y': \01^{\log s} \rightarrow \01^{\log s}$ and these will serve as the functions for the $\NM$ protocol. We define Alice's functions $\{f_x'\}_x$ in the following way: for every $x\in \01^n, a_i \in \01^{\log k}, b_j \in \01^{\log \ell}$ if on the node $i\in [k]$, $f_i(x)=b_j$ then we define $f_{x}'(b_i) = a_j$. Similarly we define Bob's functions $\{g_y'\}_y$ as follows: for every $y\in \01^n, a_i \in \01^{\log k}, b_j \in \01^{\log \ell}$ if on the node $j\in [\ell]$, $g_j(y)=a_i$ then we define $g_{y}'(a_j) = b_i$.  From the construction of the $\GBBP$, it is not too hard to see that the number of different messages used in the $\NM$ protocol is at most  the number of nodes in the bi-partitions. Hence, each message size at most $\log s$.
\end{proof}

	Earlier we saw that $\GBBP$ is polynomially related to $\BBP$. We now observe that both these measures can be exponentially smaller than standard branching program size.\footnote{The function we use here is the standard function that separates bipartite formula size from formula size.}
	
\begin{fact}
The parity function $\parity_n(x,y)=\sum_{i}x_i \oplus y_i \pmod{2}$ gives an exponential separation between generalized bipartite branching programs and branching programs.
\end{fact}

\begin{proof}
In order to see that bipartite branching  program can compute $\parity_n$ efficiently, we use the characterization in Theorem~\ref{thm:d_nomem_logT}: in the memoryless communication model, Alice can simply compute $b_1=\sum_i x_i \pmod{2}$ and send it to Bob who computes $b_2=\sum_i y_i \pmod{2}$ and outputs $b_1\oplus b_2$. However, using Ne\v{c}hiporuk~\cite{nechiporuk:cc} it is possible to show that an arbitrary branching program computing $\parity_n$ is at least $n$. The idea of the proof is that setting any variable $x_i$ to a constant $b$, would give us two different functions on the rest of the bits depending on $b$: either $\parity_{n-1}$ or $1\oplus \parity_{n-1}$. Thus fixing each bit gives us two different sub-functions and using Ne\v{c}hiporuk we obtain a lower bound of $n$ on the branching program size.
\end{proof}

\paragraph{Time Space Trade-off for Memoryless.}  Finally, we mention a connection between our communication model and time-space trade-offs. In particular, what are the functions that can be computed if we limit the number of rounds in the memoryless protocol? Earlier we saw that, an arbitrary memoryless protocol of cost $s$ for computing a function $F$ could consist of at most $2^{s+1}$ rounds of message exchanges. If sending one message takes one unit of time, we can ask whether it is possible to simultaneously reduce the message size $s$ and the time $t$ required to compute a function. The fact below gives a time-space trade-off in terms of deterministic communication~complexity.

\begin{fact}
	    For every $k\geq 1$ and  function $F:
	    \01^n\times \01^n\rightarrow \01$, we have $\NM_k(F) \geq \textsf{D}(F)/k$, where $\NM_k(F)$ is the $\NM$ communication complexity of computing $F$ with at most $k$ rounds of communication, and $\textsf{D}(F)$ is the standard deterministic communication complexity.
\end{fact}
\begin{proof}Observe that an arbitrary deterministic communication protocol for computing any function $F$ can be obtained by simulating a memoryless protocol of cost $\NM_k(F)=s$: at each of the $k$ rounds, Alice or Bob send an $s$-bit message to one another. Hence $\textsf{D}(F)\leq \NM_k(F)\cdot k$.
\end{proof}

It is not hard to now see that the number of rounds in an $\NM(F)$ protocol corresponds to the \emph{depth} of the generalized bipartite branching program computing $F$. So an immediate corollary of the fact above is, even for simple functions such as equality, inner product, if we restrict the depth of $\GBBP$ to be $o(n)$, then we can show \emph{exponential-size} lower bounds on such $\GBBP$s computing these functions. Similarly note that one can separate $\QNM$ and $\NM$ model of communication if we bound the number of rounds: consider the problem where Alice and Bob get $x,y\in \01^n$ respectively promised that, $x=y$ or Hamming distance between $x,y$ is $n/2$. In this case, clearly $\NM_k(F)\geq n/k$ (from the fact above), which in particular means that constant-round $\NM$ protocols need to send $\Omega(n)$ bits. In contrast, in the $\QNM$ model, Alice could simply send $O(1)$ copies of a fingerprint state $\ket{\psi_x}=\frac{1}{\sqrt{n}}\sum_i (-1)^{x_i}\ket{i}$ (in a \emph{single round}) and due to the promise, Bob can perform swap test between $\ket{\psi_x}, \ket{\psi_y}$ and decide if $x=y$ or the Hamming distance is $n/2$ with probability $1$.

\section{Separations between memoryless communication models}
\label{sec:understandmem}

	In this section, we show that there exists exponential separations between the four memoryless communication models defined in Section~\ref{sec:memcommmode} (and in particular, Fact~\ref{claim:firstinequality}).
	\begin{theorem}
	\label{thm:exponentiallyweakinequality}
		There exists functions $F$ for which the following inequalities (as shown in Fact~\ref{claim:firstinequality}) is exponentially weak\footnote{We remark that the functions exhibiting these exponential separations are different for the three inequalities.}
		$$
		\M(F)\leq 	\NM(F)  \leq 2\M^{\rightarrow}(F)\leq	2\NM^{\rightarrow}(F).
		$$
	\end{theorem}
	
	\begin{proof}
		The third inequality is exponentially weak for the Disjointness function $\DISJ_{n}$ defined as: for every $x,y\in \01^n$,  $\DISJ(x,y)=0$ if and only if there exists $i\in [n]$ such that $x_i=y_i=1$. In this model, we have $\NM^{\rightarrow}(F)\geq \Omega(n)$. In order to see this, we saw in Fact~\ref{fact:nm-oneway_equal_det_one_way} that $\NM^{\rightarrow}(F)=D^{\rightarrow}(F)$ which is just the standard one-way communication model. It is well known~\cite{ccbook} that  $D^{\rightarrow}(\DISJ_n)= \Omega(n)$. However, when Alice has a memory they can perform the following protocol: Alice sends an index with the value $(i,x_i)$ which takes $\log n+1$ bits and Bob gets a symbol in $\{0,\perp\}$:~$0$ if $x_i=y_i=1$ and $\perp$ if $x_i\neq y_i$. Bob outputs 0 if he gets 0, otherwise for $\perp$ Bob continues. Note that Alice doesn't know Bob's output since it's a one way $M^{\rightarrow}(F)$ protocol. Regardless of Bob's output, Alice repeats the protocol above for $n$ different $(i,x_i)$ (this is where we use the fact that Alice has a memory, hence she doesn't repeat sending the same $i$ twice). At the $n$th round after Alice sends $(n,x_n)$, she sends an one-bit `output' all-$1$ string. Note that if Bob didn't output~$0$ in the first $n$ rounds, he will output $1$ once he gets the all-$1$ message. 
		
		The second inequality is exponentially weak for the inner-product function $\IP_{n}$, defined as  
		$\IP(x,y)=\sum_{i} x_i\oplus y_i\pmod{2}$ for $x,y\in \01^n$. Papakonstantinou et al.~\cite{memoryless:comm}  showed that $\M^{\rightarrow}(\IP_n)\geq n/4$. In the memoryless communication model, we showed $\NM(\IP_n)\leq \log n+1$ at the start of this section.
	
The first inequality is weak for a random function. Let $F:\01^n\times \01^n\rightarrow \01$ be a random function, then we showed in Lemma~\ref{lemma:randomfhasighR} that $\NM(F)\geq \Omega(n)$.	Also  $\M(F)\leq \log n+1$ follows immediately from Fact~\ref{fact:MQMissmall}.
	\end{proof}
	
	We now exhibit exponential separations between the quantum and classical memoryless models of communication complexity.
	
	\begin{theorem}
	\label{thm:qcseparations}
		There exist functions $F:D\rightarrow \01$ where $D\subseteq \01^n\times \01^n$ for which the following inequalities are exponentially weak: (i) $\QNM^{\rightarrow}(F)\leq \NM^{\rightarrow}(F)$, (ii) $\QM^{\rightarrow}(F)\leq \M^{\rightarrow}(F)$, (iii) $\QM(F)\leq \M(F)$.\footnote{Again, the functions exhibiting these separations are different for the three inequalities.}
	\end{theorem}

	\begin{proof} We prove the three separations individually below.
	
	\paragraph{First inequality.} In order to prove that the first inequality is exponentially weak, we use the the standard Boolean Hidden matching problem introduced in~\cite{yossef:exponential,gavinsky:exponential}: Alice is given $x\in \01^n$, Bob is given a matching $M\in \01^{n/2\times n}$ on an $n$-vertex bipartite graph\footnote{Think of $M$ as the incident matrix where the rows are labelled by the $n/2$ edges and the columns are indexed by the $n$ vertices. Since $M$ is a matching every row has Hamming weight at most $2$.} and $y\in \01^n$ and they need to decide if $Mx=y$ or $Mx=y\oplus 1^n$. 	The quantum protocol in the memoryless setting is the following: Alice prepares  the $(\log n)$-qubit state $\ket{\psi}=\frac{1}{\sqrt{n}}\sum_{x\in \01^n}(-1)^{x_i}\ket{i}$ and sends it to Bob.  Bob performs the two-outcome measurement $\big\{\frac{1}{\sqrt{2}}(\ket{k}\pm\ket{q}):(k,q)\in M\big\}$ where $(k,q)\in [n]^2$ refers to the two vertices of of an edge in the matching. Observe that probability of obtaining a basis state outcome $\frac{1}{\sqrt{2}}(\ket{k}+\ket{q})$ is given~by
		$$
		\frac{1}{\sqrt{2}}\langle \psi\vert (\ket{k}+\ket{q}) \rangle =\frac{1}{2n}((-1)^{x_k}+(-1)^{x_q})^2,
		$$
		which equals $0$ if $x_k\oplus x_q=1$. So if Bob obtains $\frac{1}{\sqrt{2}}(\ket{k}+\ket{q})$, he knows with certainty that $x_k\oplus x_q=0$ and similarly if he obtains the state $\frac{1}{\sqrt{2}}(\ket{k}-\ket{q})$ he is sure that $x_k\oplus x_q=1$. Now, Bob looks at which row in the matching corresponds to the edge $(k,q)$ and suppose it is the $i$th row, then Bob outputs $x_k\oplus x_q\oplus y_i$. Note that this bit equals $0$ if and only if  $Mx=y$ and is $1$ otherwise. Hence, $\QNM^{\rightarrow}(F)\leq O(\log n)$.
		In order to prove the memoryless lower bound, we immediately get $\NM^{\rightarrow}(F)=D^\rightarrow(F)=\Omega(\sqrt{n})$, where the first equality follows from Fact~\ref{fact:nm-oneway_equal_det_one_way} and the second equality was proven in~\cite{gavinsky:exponential}.
		
        \paragraph{Second inequality.}  In order to prove that the second inequality is exponentially weak, we first define a partial function \emph{gap-hamming distance} $\GHD_n:D\rightarrow \01$ where $D=\{(x,y)\in \01^{2n}: d(x,y)\leq n/3 \text{ or } d(x,y)\geq 2n/3\}$ and $\GHD(x,y)=0$ if $d(x,y)\leq n/3$ and $1$ if $d(x,y)\geq 2n/3$. We then define the total function~$g$ as $g(x,y)=\GHD(x,y)$ if $(x,y)$ satisfy the promise of $\GHD$  and otherwise $g(x,y)=0$. In order to compute $g$, Alice first sends $O(1)$ copies the state $\ket{\psi_x}=\frac{1}{\sqrt{n}}\sum_i(-1)^{x_i}\ket{i}$. Bob first performs the swap test between $\ket{\psi_x}$ and $\ket{\psi_y}=\frac{1}{\sqrt{n}}\sum_i(-1)^{y_i}\ket{i}$. The swap test~\cite{swaptest} is a well-known quantum protocol takes in two quantum states $\ket{\phi},\ket{\psi}$ as inputs and outputs a bit $b\in \01$ such that 
        $$
        \Pr[b=1]=1/2+|\langle \phi | \psi \rangle|^2/2.
        $$
        In our quantum protocol, the swap test outputs~$1$ with~probability 
		$$
		\Pr[1]=\frac{1}{2}+\frac{1}{2}|\langle \psi_x | \psi_y \rangle|^2=	\frac{1}{2}+\frac{1}{2}\Big(\frac{1}{n}\sum_i (-1)^{x_i\oplus y_i}\Big)^2. 
		$$
		In the case where $d(x,y)\leq n/3$, observe that $\Pr[0]\geq 2/3$ and in the case where $d(x,y)\geq 2n/3$ we have $\Pr[1]\geq 2/3$. So Bob runs $O(1)$ swap tests and suppose he obtains $\geq 2/3$-fraction of $1$-outcomes, he outputs $1$ and suppose he obtains $\geq 2/3$-fraction of $0$-outcomes, he outputs $0$. If this were not the case, then Bob simply outputs $0$. 
		Observe that with constant error probability (which can be reduced by performing more swap tests) Bob's output equals $g(x,y)$ for all $x,y$. In particular, the overall communication cost to compute $g$ was $O(\log n)$, so we have $\QM^{\rightarrow}(g)=O(\log n)$. However, Song~\cite{song2014space} (in particular, Theorem 4.11 in his PhD thesis) shows that \emph{every} total-function extension $g$ of the  $\GHD_n$ problem satisfies $\M^{\rightarrow}(g)\geq \Omega(n)$.\footnote{Here total-function extension of $\GHD_n$ means that: for every $g:\01^n\times \01^n\rightarrow \01$ defined as $g(x,y)=\GHD_n(x,y)$ if $d(x,y)\leq n/3$ or $d(x,y)\geq 2n/3$ and $g(x,y)$ is defined arbitrarily to take values in $\01$ when $d(x,y)\in \{n/3,\ldots,2n/3\}$.} Hence $g$ gives an exponential separation for the second inequality.
		
		\paragraph{Third inequality.} 
         For the last inequality we use a promise distance $\Nequality$ function defined as $\Nequality(x,y)=1$ if $|\Int(x)-\Int(y)|\leq \sqrt{N}$, and $\Nequality(x,y)=0$ for  $(x,y)$ satisfying $|\Int(x)-\Int(y)|\geq 2\sqrt{N}$, promised the inputs to $\Nequality$ satisfy one of these two conditions. We first show that for $\QM(\Nequality)= 1$. 
  		Here, Bob sends the same $1$-qubit state $\sqrt{1-\textsf{Int}(y)/N}\ket{0}+\sqrt{\textsf{Int}(y)/N}\ket{1}$ for $T=\Theta(2^{3n})$ many rounds (although Bob doesn't have the memory to store which round he is in, Alice can send him bits $\in \01$ asking him to send these $T$ qubits). Alice stores each of the qubits, and after receiving $T$ such qubits she measures all these qubits. Suppose she obtained the string  $\widetilde{y}\in \01^T$, it is not hard to see by a Hoeffding bound that, with high probability, $\textsf{Int}(\widetilde{y})$ is close to $\textsf{Int}(y)$.\footnote{Measuring this qubit $T$ times is equivalent to flipping $T$ $\in \01$-valued coins where the probability of $1$ is $\Int(y)/N$. Suppose we flip $T$ coins and observe the number of $1$s in these flips, let us call this number $M$. Observe that expectation of $M=T\cdot \Int(y)/N$, then the multiplicative Chernoff bound implies that $\Pr[M\leq(1-\delta)T\cdot \Int(y)/N]\leq \exp(-\delta^2T\cdot \Int(y)/2N)$ and $\Pr[M\geq(1+\delta)T\cdot \Int(y)/N]\leq \exp(-\delta^2T\cdot \Int(y)/2N)$. By letting $T=\Theta(N/\delta^2)$ and choosing $\delta=1/N$, we can ensure that $M\cdot N/T$ is $(1-O(1/N))$-close to $\Int(y)$.} Hence with success probability arbitrarily close to $1$, Alice obtains a $\widetilde{y}\in \01^n$ such that $|\Int(\widetilde{y})-\Int(y)|\leq O(1)$. Once Alice obtains such a strong approximation of $y$, she can simply check if $|\Int(x)-\Int(\widetilde{y})|$ is at most $\sqrt{N}$ or at least $2\sqrt{N}$ and compute $\Nequality(x,y)$.   

 	We now show that $\M(\Nequality)\geq \Omega(\log n)$. To see this, we first show that $\M(f) \geq \Omega(\log \textsf{D}(F))$ (where $\textsf{D}(F)$ is the standard deterministic communication complexity). Suppose $\M(\Nequality)=s$, then observe that the number of rounds in this protocol is at most $2^s$: note that Alice has memory and Bob has no memory, so there is no point in Alice sending the same message $m\in \01^s$ twice, on both instances Bob would return $g_y(m)$ and Alice would have known $g_y(m)$ the first time she sent $m$ to Bob. So after $2^s\cdot s$ bits of communication, this entire transcript is a perfectly fine transcript for a \emph{detereministic} communication protocol, so we have $\textsf{D}(F)\leq 2^s\cdot s$.  It now remains to show that $\textsf{D}(F)$ is large.  In order to show this lower bound, let us fix Alice's input and for simplicity we can let $x=0^n$, in which case they need to decide if $\textsf{Int}(y)\leq \sqrt{N}$ or $\textsf{Int}(y)\geq 2\sqrt{N}$. The only way this is possible is if at most  the rightmost $(\log \sqrt{N}) - 1$ bits of $y$ are one, or if there exists at least a single $1$ in the last $n-\log \sqrt{N}$ bits of $y$. It is easy to see that an adversarial argument shows that the deterministic communication of deciding which  is the case, takes $\log \sqrt{N}=\Omega(\log n)$ bits of communication between Alice and Bob. This implies  $\textsf{D}(\Nequality)\geq \Omega(\log n)$, hence we have $\M(\Nequality)=\Omega(\log \log n)$.  	
 			\end{proof}

One drawback in the exponential separations above is that we allow a quantum protocol to err with constant probability but require the classical protocols to be correct with probability $1$.  We remark that except the second inequality, the remaining inequalities also show exponential separations between the randomized memoryless model (wherein Alice and Bob have public randomness and are allowed to err in computing the function) versus the corresponding quantum memoryless model. A natural question is to extend these separations even when the classical model is allowed to err with probability at least $1/3$.

	\subsection{Relating the Garden-hose model, Space-bounded communication complexity and Memoryless complexity}
	\label{sec:garden}
	In this section, we show that the memoryless communication complexity $\NM(F)$ of a Boolean function~$F$ is equal to the logarithm of the garden-hose complexity up to an additive constant and is equal to the space bounded communication complexity up to factor 2. Thus memoryless communication complexity $\NM$ provides a clean and simple framework to study these models.
	The garden-hose model of computation was introduced by Buhrman et al.~\cite{gardenhose:intro} to understand quantum attacks on position-based cryptographic schemes. The space bounded communication complexity was introduced by Brody et al.~\cite{brody2013space} to study the effects on communication complexity when they players are limited in their ability to store information from previous rounds. 
	Below we briefly define the garden-hose model and the space-bounded communication complexity model.
	
		\paragraph{Garden-hose model.} In the garden-hose model of computation, Alice and Bob are neighbours {(who cannot communicate)} and have few pipes going across the boundary of their gardens. Based on their private inputs $x,y$ and a function $F:\01^n\times \01^n\rightarrow \01$ known to both, the players connect some of the opening of the pipes on their respective sides with garden-hoses. Additionally, Alice connects a tap to one of the pipes on her side. Naturally, based on the garden-hose connections,  water travels back and forth through some of the pipes and finally spills on either Alice's or Bob's side, based on which they decide if a function $F$ on input $x,y$ evaluates to $0$ or $1$. It is easy to show that Alice and Bob can compute every function using this game. The garden-hose complexity $\GH(F)$ is defined to be the minimum \emph{number of pipes} required to compute $F$ this way for all possible inputs $x,y$ to  Alice and Bob. For more on garden hose complexity, we refer the interested reader to~\cite{gardenhose:intro, podder:garden,speelman:thesis2,speelman:thesis}.

        \paragraph{Space bounded communication complexity} Alice and Bob each have at most $s(n)$ bits of memory. Based on their private inputs $x,y$ they want to compute the function $F:\01^n \times \01^n \rightarrow \01$ in the following manner: At each round Alice receives a single bit message $m_B \in \01$ from Bob and based on her input $x$, the incoming message $m_B$ and her previous $s(n)$-bit register content, she computes a new $s(n)$-bit register  and decides whether to stop and output $0/1$ or to continue. Bob does the same. At the beginning of the game, the register contents of both players are set to the all-zero strings. The game then starts by Alice sending the first message and continues until one of players outputs 0/1. Space bounded communication complexity $\SM(F)$ of computing a function~$F$ is the minimum register size $s(n)$ required to compute $F$ on the worst possible input~$(x,y)$.    
        
        They have also studied one-way communication complexity variant of this space bounded memory model in which Bob can have two types of memory: an oblivious memory (Which gets updated depending on only Alice’s message) and a non-oblivious memory (which can depend on Bob's input).
                Brody et al.~\cite{brody2013space} claimed that space bounded communication complexity is equal to the garden-hose communication complexity upto factor 2.
        
        \begin{claim}[\cite{brody2013space}]
        \label{label:claim-brody-gh-sm}
        For every function $F$ there exists constants $c\in (0,1), d \in \mathbb{N}^+$ such that \[ c \cdot 2^{\SM(F)} \leq \GH(F) \leq 2^{2\SM(F)+2} +d \]
        \end{claim}

    We show the following relation between 
    
    \begin{lemma}
    \label{label:lemma-nm-sm}
    For every function $F$,
    \[
    \NM(F) \leq 2\SM(F)+1, \text{ and } \hspace{2mm} \SM(F) \leq \NM(F) + \log \NM(F)
    \]
    \end{lemma}
	
	\begin{proof} 
	We first prove the first inequality. Let $\SM(F)=t$. In the protocol, let Alice's and Bob's registers be $m_A,m_B \in \01^{t}$ respectively. Then for the memoryless protocol, Alice and Bob can put $m_A,m_B$ as the message and send them back and forth along with the intended one bit message that they would have  send in any round of the space bounded protocol. Thus a message in the $\NM$ protocol would be a tuple $(m_A,m_B,\01)$. Hence we get the desired $2\SM(F)+1$ bound.
	
	For the second inequality, let $\NM(F)=t$. In any space bounded protocol, since the players are allowed to have local registers but the message size is a single bit, simulating a $\NM$ protocol can be done in the following way: at each round the players store the $t$-bit long message (that they would have sent in the $\NM$ protocol) in their local memory and then send it bit by bit. Since Alice and Bob need to know which bit of the $t$ bits message are being sent at any round they both maintain a $\log t$ bits clock as well. Thus any $\NM$ protocol with $t$ message size can be simulated by a space bounded protocol with $t+ \log t$ bits message local registers.   
	\end{proof}
	
	Using the Claim~\ref{label:claim-brody-gh-sm} and Lemma~\ref{label:lemma-nm-sm} we can conclude that the logarithm of the garden-hose complexity is equal to the memoryless $\NM$ complexity up to factor 2. This seems interesting already given we can connect these two models, but in the $\NM$ model, even factor-$2$s are important since they are related to formula lower bounds. 
	Now we show that it is possible to further tighten the relation in the lemma above. Below we show that $\NM$ is actually equivalent to the logarithm of the garden-hose complexity up to an \emph{additive term} of 4.
	The first observation relating the garden-hose model and memoryless communication complexity is that, the garden-hose model is exactly the $\NM$ communication model, except that in addition to the memoryless-ness of Alice and Bob, there is a \emph{bijection} between the incoming and the outgoing messages of both players (i.e., the local functions Alice and Bob apply $\{f_x:\01^s\rightarrow \01^s\}_x,\{g_y:\01^s\rightarrow \01^s\}_y$  are \emph{bijective functions}. We now state and prove the theorem which shows how $\GH$ is related to the standard memoryless communication model. We thanks Florian Speelman for the proof of the first inequality below. 
	
	\begin{theorem}\label{lem:nmghm}
	For $F:\01^n\times \01^n\rightarrow \01$, we have $\log\GH(F) -4 \leq \NM(F)\leq \log \GH(F)$.
	\end{theorem}
	
	\begin{proof}
	We first prove the first inequality. The idea is to use the construction of Lange, McKenzie and Tapp~\cite{lange2000reversible} to make the computations reversible. For every input pair $(x,y)$, consider the directed graph that determines the $\NM$ protocol with communication $s$. The $2\cdot 2^s$ nodes are labeled by $A_{i}$, for $i \in [2^s]$,  and the same
for the different $B_j$ with $j \in [2^s]$. These nodes represent the message that Alice, respectively Bob, has just received.

Since for every $(x,y)$ Alice and Bob compute fixed functions $f_x, g_y$ respectively, for every node we imagine there being an outgoing edge if (for the
current input) receiving that message would cause the receiving player
to send a specific outgoing message. So $A_i$ has an outgoing edge towards~$B_j$ if upon receiving message $i$, Alice would reply with the message $j$, i.e., $f_x(i)=j$.
In other words, the single outgoing edge of $A_i$ (if any) is determined by $f_x$, and the outgoing edge of $B_j$ is determined by $g_y$. Say for
simplicity that the protocol starts with $B_0$, and there are two special
unique 0 and 1 nodes that determine the output. Also imagine that only
edges from $A$ nodes go to 0, and edges from $B$ nodes go to 1.  The size of
this graph is at most $2^{s+1}$ (we don't count the 0/1 nodes since we
won't need states to simulate them in the $\GH$ model).

Thus the outdegree of every node of this graph is 1, but the indegree is not, and that creates the problem with simulating it in the $\GH$ model. But this can be overcome by using the reversible simulation tricks of \cite{lange2000reversible}. In our simulation, instead of following the path of the original execution, we're going to perform a walk along the connected component that involves the starting
state. This path might visit all nodes in this component (so it is not
time-efficient scheme) but it is one-to-one.
For every node $A_i$, in the $\GH$ protocol we have two pipes. Let us call them $A_i^{f}$ (short for, $A_i$ forward) and $A_i^{r}$ (short for, $A_i$ reverse), and similarly for the $B_j$ nodes.
Alice then has the following strategy for connecting the pipes:
Consider all $i$ in lexicographic order. For each node $A_i^{f}$, Alice looks at which message $j$ she would send if receiving $i$ (i.e., $f_x(i)=j$), and she connects $A_i^{f}$ with the corresponding $B_j^{f}$ in the following manner:
\begin{itemize}
    \item If that connection is already filled, since an earlier $i'$ has already been connected there, Alice connects $A_i^{f}$ to $A_{i'}^{r}$ instead.
    \item Additionally, if for the current $i$ there is no later $i'$ such that $i'$ also tries to connect to $j$, connect $B_j^{r}$ to $A_i^{r}$.
\end{itemize}
Alice also considers all $j$ in lexicographic order. She already has some connection to both~$B_j^{r}$ and~$B_j^{f}$ if there exists an
incoming message $i$ so that Alice would send the message $j$ to Bob. 
If $B_j^{f}$ and $B_j^{r}$ are still empty, Alice connects them together.
Whenever Alice's arrow would go towards 0, she just keeps the pipe open.
We imagine the water originally coming out of $A_i^{f}$, so we don't
need a pipe with that label. Bob's strategy is exactly the same as Alice's, replacing the role of $A_i$ with $B_j$ and $x, y$. And whenever Bob would connect some $B_j^{f}$ to 1 he can just keep it open.

Now if we follow the path along the graph starting from $B_0$, whenever we reach any node $p$ with multiple incoming edges, and we're not the lexicographical first node that connects to it, we're going to go 'backwards' via a sequence $A^{r}$ and $B^{r}$ pipes. But at some point, this backtracking will stop, by encountering a node without predecessor, and we'll come back to the node $p$ again using the corresponding forward pipes. After many such rounds of backtracking we will go forward using the node $p$. Note that Alice and Bob can make the above connections individually only based on their private functions $f_x, g_y$. In total we have multiplied the number of states with 4, which will
become an additive 4 in the comparison with $\NM(F)$.

	Next we prove the second inequality. Let $x,y\in \01^n$ and $\GH(F,x,y)=s$, i.e., there exists~$s$ pipes that Alice and Bob place which form the garden-hose protocol to compute $F(x,y)$. Let the pipe openings on Alice's side be indexed by $A_1,\ldots,A_s$ and  Bob's side be $B_1,\ldots,B_s$.  We now define an $\NM$ protocol as follows: let $f_x:\01^{\log s}\rightarrow \01^{\log s}$ and $g_y:\01^{\log s}\rightarrow \01^{\log s}$ be defined as, for $u,v\in [s]$, we $f_x(u)=v$ if and only if there exists a pipe between $A_u$ and $A_v$ on Alice's side. Similarly, for Bob's side we let $g_y(w)=z$ iff $B_w$ is connected to $B_z$ for $w,z\in [s]$. So, in the $\NM$ protocol, Alice and Bob exchange 
the name of the pipes (which takes $\log s$ bits) through which water flows back and forth in the garden-hose protocol. This continues until one of the players sees a pipe opening in which case Alice or Bob output either $0$ or $1$ respectively. It is clear that the message size is at most $\log s$ since it takes at most $\log s$ bits to provide the name of one of the $s$ pipes. One can easily observe that as long as the garden-hose protocol computes $F$, the $\NM$ protocol also computes~$F$. 
An alternate way to look at this proof of second inequality is via the logarithm space Turing machine characterization of the garden-hose model \cite{gardenhose:intro}. A garden-hose game for computing a function $F$ is polynomially equivalent to a logspace Turing machine $M_F$, which given the pre-determined functions $g,h$, such that $F(x,y) = f(g(x),h(y))$, computes $f$ (The logspace Turing machine for $f$ just alternately looks in the $\NM$ strategy of Alice and Bob and considers which message to send). 
	\end{proof}

Interestingly, Theorem~\ref{lem:nmghm} together with Theorem~\ref{thm:exponentiallyweakinequality} gives us a way to construct a garden-hose protocol using an $\M^\rightarrow$ protocol and, as we will see below, this could result in potentially stronger upper bound on the garden-hose model. In an earlier work of Klauck and Podder~\cite{podder:garden}, it was conjectured that the disjointness function with input size $m=n \cdot 2\log n$ (i.e., with set size $n$ and universe size $n^2$) has a quadratic lower bound $\Omega(m^2)$ in the garden-hose model. Here, we show that $\GH$ protocol for this problem has cost $O(m^2 / \log^2 m)$. Although the improvement is only by a logarithmic-factor, we believe that this complexity can be reduced further which we leave as an open question. 

\textbf{Disjointness with quadratic universe:} Alice and Bob are given $n$ numbers each from $[n^2]$ as a $m=n \cdot 2\log n$ long bit strings. Their goal is to check if all of their $2n$ numbers are unique. Without loss we can assume that the $n$ numbers on the respective sides of Alice and Bob are unique, if not they can check it locally and output $0$ without any communication.  Then an $\M^\rightarrow$ protocol for computing this function is as follows: Alice keeps sending all her numbers to Bob one by one (using her local memory to keep track of which numbers she has already sent). This requires $2\log n$ size message register on every round. Bob upon receiving any number from Alice, checks if any number of his side matches the number received. If there is a match he outputs $0$, else he continues. For the last message Alice sends the number along with a special marker. Bob performs his usual check and output 1 if the check passes and the marker is present. Clearly the cost of this protocol is $2 \log n$ and thus from Theorem~\ref{lem:nmghm} the garden-hose protocol for computing this function has cost $n^2$. Since the input size is $m=n \cdot 2\log n$, the cost of the garden-hose protocol is $O(m^2/\log^2 m)$.

% \vspace{0mm}
 \paragraph{Acknowledgements.} 
We thank Florian Speelman for the proof of the first inequality of Lemma~\ref{lem:nmghm} and letting us include it in our paper. 
 We would like to thank Mika G{\"o}{\"o}s, Robert Robere, Dave Touchette and Henry Yuen for helpful pointers. SP also thanks Anne Broadbent and Alexander Kerzner for discussions during the early phase of this project as a part of discussing garden-hose~model.  
We also thank anonymous reviewers for the helpful comments on the first draft of this work. 
 Part of this work was done when SA was visiting University of Ottawa (hosted by Anne Broadbent) and University of Toronto (hosted by Henry Yuen) and thank them for their hospitality. SA was supported in part by the Army Research Laboratory and the Army Research Office under grant number W911NF-20-1-0014.

\newcommand{\etalchar}[1]{$^{#1}$}

\end{document}